\documentclass[conference,10pt]{IEEEtran}

\usepackage{booktabs}   
\usepackage{subcaption} 
\usepackage{yfonts,array}
\usepackage{amssymb,amsfonts,amsmath,amsthm,units,nicefrac}
\usepackage{upgreek,xcolor,bm,tikz,stackrel}
\usepackage[inline]{enumitem}
\usepackage{graphicx}
\usepackage[all]{xy}
\usepackage{tikz-cd}
\usepackage{multicol}
\usepackage{color}
\usepackage[inline]{enumitem}
\usetikzlibrary{arrows}

\newcommand{\rel}{\sqsubseteq}

\newcommand{\ler}{\sqsupseteq}

\newcommand{\iiff}{\leftrightarrow}
\newcommand{\ineg}{\neg}
\newcommand{\imp}{\to}

\newcommand{\lanfull}{{\mathcal L}_{\ps\nec}}

\newcommand{\seq}{\succcurlyeq}
\newcommand{\lgt}[1]{\|#1\|}

\newcommand{\nec}{\Box}
\newcommand{\ps}{\Diamond}

\newcommand{\peq}{\preccurlyeq}

\def\<{\left (}

\def\>{\right )}
\def\({\left (}
\def\){\right )}

\DeclareSymbolFont{AMSb}{U}{msb}{m}{n}
\DeclareMathSymbol{\N}{\mathbin}{AMSb}{"4E}
\DeclareMathSymbol{\Z}{\mathbin}{AMSb}{"5A}
\DeclareMathSymbol{\R}{\mathbin}{AMSb}{"52}
\DeclareMathSymbol{\Q}{\mathbin}{AMSb}{"51}
\DeclareMathSymbol{\I}{\mathbin}{AMSb}{"49}
\newcommand{\fallible}[1]{#1^\bot}
\newtheorem{thm}{Theorem}[section]
\newtheorem{defn}[thm]{Definition}
\newtheorem{lem}[thm]{Lemma}
\newtheorem{prop}[thm]{Proposition}

\newtheorem{remark}[thm]{Remark}

\makeatletter
\newcommand{\mlabel}[2]{#2\def\@currentlabel{#2}\label{#1}}
\makeatother

\newcommand{\eqdef}{\ensuremath{\mathbin{\raisebox{-1pt}[-3pt][0pt]{$\stackrel{\mathit{def}}{=}$}}}}

\begin{document}

\title{Some constructive variants of S4 with the finite model property}

\author{\IEEEauthorblockN{Philippe Balbiani\IEEEauthorrefmark{1},
Mart\'in Di\'eguez\IEEEauthorrefmark{2} and
David Fern\'andez-Duque\IEEEauthorrefmark{3}}
\IEEEauthorblockA{\IEEEauthorrefmark{1}IRIT, Toulouse University, Toulouse, France\\ Email: philippe.balbiani@irit.fr}
\IEEEauthorblockA{\IEEEauthorrefmark{2}LERIA, University of Angers, Angers, France\\
Email: martin.dieguezlodeiro@univ-angers.fr}
\IEEEauthorblockA{\IEEEauthorrefmark{3}Department of Mathematics, Ghent University, Ghent, Belgium\\ Email: david.fernandezduque@UGent.be}}

\IEEEoverridecommandlockouts
\IEEEpubid{\makebox[\columnwidth]{978-1-6654-4895-6/21/\$31.00~\copyright2021 IEEE \hfill} \hspace{\columnsep}\makebox[\columnwidth]{ }}

\maketitle
\begin{abstract}
The logics $\sf CS4$ and $\sf IS4$ are intuitionistic variants of the modal logic $\sf S4$.
Whether the finite model property holds for each of these logics has been a long-standing open problem. In this paper we introduce two logics closely related to $\sf IS4$: $\sf GS4$, obtained by adding the G\"odel–Dummett axiom to $\sf IS4$, and $\sf S4I$, obtained by reversing the roles of the modal and intuitionistic relations. We then prove that $\sf CS4$, $\sf GS4$, and $\sf S4I$ all enjoy the finite model property.
\end{abstract}

\section{Introduction}

When extending intuitionistic logic with modal operators, the absence of the {\em excluded middle axiom} raises several variants of the distributivity axiom $\bf K$ leading to, mainly, two different approaches for intuitionistic-based modal logics: {\em intuitionistic} and {\em constructive} modal logics.
{\em Intuitionistic modal logics} have been studied by Plotkin and Stirling~\cite{PlotkinS86},  Fischer Servi~\cite{servi1977modal,servi1984axiomatizations} and Simpson~\cite{Simpson94}, whose aim is to define analogous of classical modalities but from an intuitionistic point of view. Within the intuitionistic settings, the $\nec$ operator (resp. $\ps$) imitates the behaviour of $\forall$ (resp. $\exists$) in first-order intuitionistic logic.

On the other hand, {\em constructive modal logics}~\cite{Bellin2001ExtendedCC} are motivated by their applications to computer science, such as the Curry–Howard correspondence~\cite{BiermanP00},  type systems for staged computation~\cite{10.1145/237721.237788} and distributed computation~\cite{moody2003:modalbasis} or even hardware verification~\cite{FAIRTLOUGH19971}. Inspired by the previous work of Fitch~\cite{fitch1948intuitionistic} and Wijesekera~\cite{wijesekera1990constructive}, those logics provide specific semantics for the $\ps$ operator. The main characteristic of this class of logics is that the addition of the Excluded Middle does not yield classical modal logic $\sf K$.

There are (at least) three prominent contstructive variants of $\sf S4$ in the literature.
This paper concerns two of them, known as $\sf IS4$ \cite{Simpson94} and $\sf CS4$ \cite{AlechinaMPR01}.
Both logics enjoy a natural axiomatization and are sound and complete for a class of Kripke structures based on two preorders, one preorder $\peq$ for the intuitionistic implication, and another preorder $\sqsubseteq$ for the modal $\nec$.
However, the finite model properties for both of these logics have remained open for twenty years or more, in the case of $\sf CS4$ since at least 2001 \cite{AlechinaMPR01}, and of $\sf IS4$ at least since 1994 \cite{Simpson94}.
In the case of $\sf IS4$, it is not even known if the validity problem is decidable.
The third is the logic ${\sf IntS4}$ studied by Wolter and Zakharyaschev, along with related variants of other modal logics~\cite{Wolter1997,Wolter1999}.
In contrast to $\sf CS4$ and $\sf IS4$, ${\sf IntS4}$ is known to enjoy the \emph{finite model property} (FMP), although the semantics are quite different, for example employing a different binary relation for each of $\ps,\nec$.

In this paper we settle the first question and prove that indeed $\sf CS4$ has the FMP. 
We also introduce two mild variants of $\sf IS4$ and show that they both enjoy the FMP.
The first logic, $\sf GS4$, is defined over a subclass of the class of $\sf IS4$ models where the intuitionistic relation is {\em locally linear,} so that it satisfies the G\"odel-Dummett axiom $(\varphi\to \psi)\vee(\psi \to \varphi)$.
The second logic, $\sf S4I$, is defined over the same class of models as $\sf IS4$, except that the roles of the intuitionistic and the modal preorders are interchanged.
The frame conditions for $\sf S4I$ are natural from a technical point of view, as they are the minimal conditions required so that both modalities $\ps$ and $\nec$ can be evaluated `classically'.

The key insight of our proof technique is that all of these logics enjoy the {\em shallow model property,} meaning that any non-valid formula $\varphi$ may be falsified in a model where the length of any $\prec$-chain is bounded (as usual, $w\prec v$ means that $w\peq v$ but $v\not\peq w$).%
\footnote{The term {\em shallow model} has been used in a similar way in the context of classical modal logic \cite{SP08}.}
While shallow models may in principle be infinite, their quotients modulo bisimulation with respect to $\peq$ are always finite.

Thus the problem of showing the finite model property is reduced to that of proving the shallow model property.
First, for each $\Lambda\in \{{\sf CS4}, {\sf GS4}, {\sf S4I}\}$, we construct a canonical model $\mathcal M^\Lambda_c = (W_c,\peq_c,\sqsubseteq_c,V_c)$ using fairly standard techniques as found in e.g.~\cite{Simpson94,AlechinaMPR01}.
Based on this canonical model, we fix a finite set of formulas $\Sigma$ and construct a shallow model $\mathcal M^\Lambda_\Sigma = (W_\Sigma,\peq_\Sigma,\sqsubseteq_\Sigma,V_\Sigma)$.
The details of the construction vary for each of the three logics we consider, but with the general theme that $w\prec_\Sigma v$ may only hold if there is some $\varphi\in \Sigma$ which holds on $v$ but not on $w$.
Having placed this restriction, it is readily seen that any chain
\[w_0 \prec_\Sigma w_1 \prec_\Sigma  \ldots \prec_\Sigma w_n\]
is witnessed by distinct formulas $\varphi_0,\ldots,\varphi_{n-1}$ of $\Sigma$ with the property that $\varphi_i$ holds on $w_{i+1}$ but not on $w_n$.
It follows that the length of the chain is bounded by $|\Sigma|+ 1$.

\paragraph*{Layout}
The layout of the paper is as follows.
In Section~\ref{SecBasic} we present the syntax, semantics and deductive calculi of the logics studied in this paper. 
In Section~\ref{secSound} we consider the soundness of the axiomatic systems we will study, while the corresponding completeness proofs are presented along sections~\ref{secCompCS4},\ref{secCompS4I} and~\ref{secCompGS4}.

The second part of this paper is devoted to the proof of the FMP of the constructive ${\sf S4}$ variants studied in this paper. To do so, we introduce $\Sigma$-bisimulations in Section~\ref{secSBisim} and Shallow models in Section~\ref{sShallow}. Those concepts are a central ingredient of the FMP proofs presented along sections~\ref{secFMPCS4},\ref{secGS4fin} and~\ref{secFMPS4I}.

We finish this paper with the conclusions and potential future lines of research.

\section{Syntax and Semantics}\label{SecBasic}

In this section we will introduce the various intuitionistic or intermediate semantics for modal logic we will be interested in.
Fix a countably infinite set $\mathbb P$ of propositional variables. Then, the {\em full (intuitionistic modal) language} $\mathcal L = \lanfull$ is defined by the grammar (in Backus-Naur form)
\[\varphi,\psi := \   p \  | \   \bot  \ |  \ \left(\varphi\wedge\psi\right) \  |  \ \left(\varphi\vee\psi\right)  \ |  \ \left(\varphi\imp \psi\right)    \  | \  \ps\varphi \  |  \ \nec\varphi   \]
where $p\in \mathbb P$.
We also use $\ineg\varphi$ as a shorthand for $\varphi\imp \bot$ and $\varphi\iiff \psi$ as a shorthand for $(\varphi\imp \psi) \wedge (\psi\imp\varphi)$.

Denote the set of subformulas of $\varphi\in \mathcal L$ by ${\mathrm{sub}}(\varphi)$, and its size by $\# {\mathrm{sub}}(\varphi)$ or $\lgt\varphi$.

\subsection{Deductive Calculi}

Next we define the deductive calculi we are interested in.

\begin{defn}
We define the logic $\sf CS4$ by adding to the set of all intuitionistic tautologies the following axioms and rules

\begin{enumerate}[itemsep=0pt]
\item[\mlabel{ax:k:box}{\ensuremath{\bm{ \mathrm K_{\nec}}}}] $\nec (\varphi \to \psi) \to(\nec \varphi \to \nec \psi)$
\item[\mlabel{ax:k:dia}{\ensuremath{\bm{\mathrm K_{\ps}}}}] $\nec (\varphi \to \psi) \to( \ps \varphi \to \ps \psi)$
\end{enumerate}
\begin{multicols}{3}
\begin{enumerate}[itemsep=2pt]
\item[\mlabel{rl:mp}{\ensuremath{\bm{\mathrm{MP}}}}] $\dfrac{\varphi \to \psi \hspace{10pt} \varphi}{\psi}$
\item[\mlabel{rl:nec}{\ensuremath{\bm{\mathrm{Nec}}}}] $\dfrac{\varphi}{\nec\varphi}$
\item[\mlabel{ax:ref:box}{\ensuremath{\bm{\mathrm T_{\nec}}}}] $\nec\varphi \to \varphi$
\item[\mlabel{ax:ref:dia}{\ensuremath{\bm{\mathrm T_{\ps}}}}] $\varphi \to \ps \varphi$
\item[\mlabel{ax:trans:box}{\ensuremath{\bm{\mathrm 4_{\nec}}}}] $\nec\varphi \to \nec \nec \varphi$
\item[\mlabel{ax:trans:dia}{\ensuremath{\bm{\mathrm 4_{\ps}}}}] $\ps \ps \varphi \to \ps \varphi$
\end{enumerate}
\end{multicols}

\noindent We define the additional axioms


\begin{multicols}{2}
\begin{itemize}
\item[\mlabel{ax:dp}{\ensuremath{\bm{\mathrm{DP}}}}] $\ps(\varphi \vee \psi) \to \ps \varphi \vee \ps \psi$;
\item[\mlabel{ax:g}{\ensuremath{\bm{\mathrm{GD}}}}] $(\varphi \to \psi) \vee (\psi \to \varphi)$;
\item[\mlabel{ax:cd}{\ensuremath{\bm{\mathrm{CD}}}}] $\nec(\varphi \vee \psi) \to \nec \varphi \vee \ps \psi$;
\item[\mlabel{ax:null}{\ensuremath{\bm{\mathrm{N}}}}] $\neg \ps \bot$.
\end{itemize}
\end{multicols}
\vspace{-12pt}
\begin{itemize}
\item[\mlabel{ax:fs}{\ensuremath{\bm{\mathrm{FS}}}}] \mbox{$\left(\ps \varphi \to \nec \psi \right) \rightarrow \nec \left(\varphi \to \psi \right)$;}
\end{itemize}
Here, \ref{ax:dp} stands for `disjunctive possibility', \ref{ax:g} for `G\"odel-Dummett', \ref{ax:cd} for `constant domain,'~\ref{ax:null} for `nullary' and~\ref{ax:fs} for `Fischer Servi'~\cite{servi1984axiomatizations}.
With this, we define the logics
\begin{align*}
{\sf IS4}& = {\sf CS4} + \ref{ax:fs} + \ref{ax:dp} + \ref{ax:null},\\[0pt]
{\sf S4I}& = {\sf CS4} + \ref{ax:dp} + \ref{ax:null} + \ref{ax:cd},\\[0pt]
{\sf GS4} & = {\sf IS4} + \ref{ax:g}.
\end{align*}
\end{defn}

Thus, $\sf CS4$ will serve as the `minimalist' logic for the purpose of this paper, and the rest of the logics we consider are extensions.
As such, it is convenient to observe that the following are already derivable in $\sf CS4$.
We leave the proofs to the reader.

\begin{prop}
The formulas

\noindent\begin{enumerate*}[itemsep=0pt,label=(\arabic*)]
	\item\label{der:ax1} $\ps \left( \varphi \rightarrow \psi\right) \rightarrow \left(\nec \varphi \rightarrow \ps \psi\right)$;	
	\item \label{der:ax3} $\ps \left( \varphi \wedge \psi\right) \rightarrow \ps \varphi \wedge \ps \psi$ and
	\item\label{der:ax2} \mbox{$\left(\nec \varphi \vee \nec \psi \right) \rightarrow \nec \left(\varphi \vee \psi\right)$}	
\end{enumerate*}
\noindent are derivable in ${\sf CS4}$.

%

\end{prop}

\subsection{Semantics}\label{sec:semantics}

We will consider several semantics leading to intuitionistic variants of $\sf S4$.
It will be convenient to introduce a general class of structures which includes all of these semantics as special cases.

\begin{defn}\label{DefSem}
An {\em intuitionistic frame} is a triple $\mathcal F=(W,\fallible W ,\peq )$, where $W$ is a set, $\peq $ is a preorder (i.e., reflexive and transitive binary relation) on $W$, and $\fallible{ W} \subseteq W$ is closed under $\peq $, i.e.~whenever $w \in \fallible{W}$ and $w \peq v$, we also have that $v \in \fallible{W}$~\cite{ArisakaDS15}.
We say that $\mathcal F$ is {\em locally linear} if $w\peq u$ and $w\peq v$ implies that $u\peq v$ or $v\peq u$.

A {\em bi-intuitionistic frame} is a quadruple $\mathcal F=(W,\fallible W ,\peq ,\rel)$, where both $ (W,\fallible W ,\peq )$ and $ (W,\fallible W ,\rel)$ are intuitionistic frames.
The bi-intuitionistic frame $\mathcal F$ is {\em locally linear} if $(W,\fallible W ,\peq)$ is locally linear.
\end{defn}

The set $\fallible W$ is called the set of {\em fallible worlds,} as in~\cite{AlechinaMPR01,ArisakaDS15}.
Note that in a bi-intuitionistic frame, $\fallible W$ is closed under both $\peq$ and $\rel$.
When $\fallible W=\varnothing$ we omit it, and view $\mathcal F$ as a triple $(W,\peq ,\rel)$.
In this case, we say that $\mathcal F$ is {\em infallible.}

Given a bi-intuitionistic frame $\mathcal F = (W,\fallible W ,\peq ,\rel)$, a {\em valuation on $\mathcal F$} is a function $V\colon \mathbb P \to 2^{W}$ which is {\em monotone} in the sense that $v\seq w \in V(p)$ implies that $v\in V(p)$ and such that $\fallible W\subseteq V(p)$ for all variables $p$ (to ensure the validity of $\bot\to \varphi$).
A {\em bi-intuitionistic model} is a structure $\mathcal M=(W, \fallible W,\peq ,\rel, V )$, consisting of a bi-frame equipped with a valuation.

We define the satisfaction relation $\models$ recursively by
\begin{itemize}	
\item $(\mathcal M,w) \models p \in\mathbb P$ if $w\in V(p)$;
\item $(\mathcal M,w) \models \bot $ if $w \in \fallible{W}$ 
\item $(\mathcal M,w) \models \varphi \wedge \psi$ if $(\mathcal M,w) \models \varphi $ and $(\mathcal M,w) \models \psi $;
\item $(\mathcal M,w) \models \varphi \vee \psi$ if $(\mathcal M,w) \models \varphi $ or $(\mathcal M,w) \models \psi $;
\item $(\mathcal M,w) \models \varphi \to \psi$ if for all $v\seq w$, $(\mathcal M,v) \models \varphi $ implies $(\mathcal M,v) \models \psi $;
\item $(\mathcal M,w) \models \ps \varphi  $ if for all $u\seq w$ there exists $v \ler u$ such that $(\mathcal M,v) \models \varphi $, and
\item $(\mathcal M,w) \models \nec \varphi $ if for all $u,v$ such that $w\peq u \rel v$, $(\mathcal M,v) \models \varphi $.
\end{itemize}

It can be easily proved that by induction on $\varphi$ that for all $w,v \in W$, if $w \peq v$ and $(\mathcal M,w)\models \varphi$ then $(\mathcal M,v)\models \varphi$.
If $\mathcal M=(W,\fallible W,\peq,\rel,V)$ is any model and $\varphi$ any formula, $\mathcal M\models\varphi$ if $(\mathcal M,w)\models \varphi$ for every $w\in W\setminus \fallible W$.
Given $\Lambda \in \lbrace {\sf CS4}, {\sf S4I }, {\sf S4I } \rbrace$, validity (in symbols $\Lambda \models \varphi$) on a frame or on a class of structures (frames or models) is then defined in the usual way.

%

Note that bi-intuitionistic models do not satisfy some of the $\sf S4 $ axioms. As an example, let us consider the bi-intuitionistic model $\mathcal M=(W,\peq ,\rel , V )$, whose corresponding frame if displayed on Figure~\ref{fig:gen-frame} and $V (p) = \lbrace x, y, z, t \rbrace$.
Note that we omit $\fallible W$: by convention, this means that $\fallible W =\varnothing$.

\begin{figure}[h!]\centering
	
	\begin{tikzpicture}[->,auto,font=\small,node distance=1.5cm]
	\node[] (x) {$x$};
	\node[] (y) [above right of = x] {$y$};
	\node[] (z) [right of = y] {$z$};
	\node[] (t) [above right of = z] {$t$};
	\node[] (w) [right of = t] {$w$};
		
	\path[->] 
	(x) edge[loop left] node{$\rel$, $\peq$}(x)
	(y) edge[loop above] node{$\rel$, $\peq$}(y)
	(z) edge[loop below] node{$\rel$, $\peq$}(z)
	(t) edge[loop below] node{$\rel$, $\peq$}(t)
	(w) edge[loop below] node{$\rel$, $\peq$}(w)
	(x) edge[] node{$\peq$}(y)
	(z) edge[] node{$\peq$}(t)
	(y) edge[] node{$\rel$}(z)
	(t) edge[] node{$\rel$}(w);
	\end{tikzpicture}
	\caption{A bi-intuitionistic frame. Transitive arrows are not displayed}
	\label{fig:gen-frame}
\end{figure}
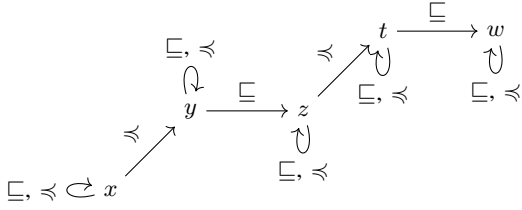		

In Figure \ref{fig:gen-frame}, it can be checked that $\mathcal M, x \not \models \nec p \to \nec \nec p$ and, therefore, $\ref{ax:trans:box}$ is not valid on the class of all bi-intuitionistic frames.
In order to make~$\ref{ax:trans:box}$ valid, we need to enforce additional constraints governing the interaction between $\peq$ and $\rel$.
There are various properties that have been used for this end.

\begin{defn}
Let $\mathcal F=(W,\fallible W,\peq)$ be an intuitionistic frame and $R\subseteq W\times W$. We say that $R$ is:
\begin{enumerate}

\item Forward confluent (for $\mathcal F$) if, whenever $w\peq w'$ and $w \mathrel R v$, there is $v'$ such that $v\peq v'$ and $w' \mathrel R v'$.

\item Backward confluent (for $\mathcal F$) if, whenever $w \mathrel R v \peq v'$, there is $w'$ such that $w \peq w' \mathrel R v'$.

\item Downward confluent (for $\mathcal F$) if, whenever $w \peq v \mathrel R v'$, there is $w'$ such that $w \mathrel R w' \peq v'$.

\end{enumerate}
A bi-intuitionistic frame $\mathcal F=(W,\fallible W,\peq,\rel)$ is forward confluent (respectively, backward confluent, downward confluent) if $\rel$ is forward confluent (respectively, backward confluent, downward confluent) for $(W,\fallible W,\peq)$.
\end{defn}

Relations satisfying the above notions of confluence have nice closure properties.
Below, if $R,S$ are binary relations then $R;S$ denotes their composition in order of application: $x \mathrel{R;S} y$ if there is $z$ such that $x\mathrel R z$ and $z\mathrel S y$.

\begin{lem}\label{lemmClosure}
Let $\mathcal F =(W,\fallible W,\peq)$ be an intuitionistic frame.
\begin{enumerate}

\item If $R,S\subseteq W\times W$ are forward (resp.~backward, downward) confluent, then so is $R;S$.

\item If $(R_i)_{i\in I} \subseteq W\times W$ are forward (resp.~backward, downward) confluent, then so is $\bigcup_{i\in I} R_i$.

\item If $R$ is forward (resp.~backward, downward) confluent, then so is its transitive closure $R^+$.

\end{enumerate}
\end{lem}

\begin{proof}
We prove only that forward confluence is closed under composition and leave other items to the reader.
Suppose that $R,S\subseteq W\times W$ are forward confluent and $v\seq w \mathrel{R;S} w'$.
Then, there is $w''$ such that $w\mathrel R w'' \mathrel S w'$.
By forward confluence of $R$, there is $v''$ such that $v\mathrel R v''\seq w $.
By forward confluence of $S$, there is $v'$ such that $v''\mathrel S v'\seq w' $.
But then, $v\mathrel{R;S} v'\seq w'$, as needed.
%
%
\end{proof}

The notions of forward and downward confluence allow us to simplify the semantic clauses for $\ps$ and $\nec$, respectively.

\begin{lem}\label{lemClassical}
Let $\mathcal M = (W,\fallible W,\prec,\sqsubseteq,V)$ be any bi-intuitionistic model, $w\in W$ and $\varphi\in\mathcal L$.
\begin{enumerate}

\item If $\mathcal M$ is forward-confluent, then $(\mathcal M,w) \models \ps\varphi$ iff $\exists v\sqsupseteq w$ such that $(\mathcal M,v) \models \varphi$.

\item If $\mathcal M$ is downward-confluent, then $(\mathcal M,w) \models \nec \varphi$ iff $\forall v\sqsupseteq w$, $(\mathcal M,v) \models \varphi$.

\end{enumerate}
\end{lem}

\begin{proof}
We prove the second claim; the first is proven similarly by dualizing.
It readily follows from the semantic clauses and the reflexivity of $\peq$ that if $(\mathcal M,w) \models \nec \varphi$ and $ v\sqsupseteq w$, then $(\mathcal M,v) \models \varphi$.
Conversely, suppose that $\forall v\sqsupseteq w, (\mathcal M,v) \models \varphi$.
Let $v\seq w$ and $v'\sqsupseteq v$.
By downward confluence, there is $w'\sqsupseteq w$ such that $w'\peq v'$.
By our assumption, $(\mathcal M,w' ) \models\varphi$.
By monotonicity of the satisfaction relation, $(\mathcal M,v' ) \models\varphi$.
Since $v,v'$ were arbitrary, we conclude that $(\mathcal M,w ) \models \nec \varphi$.
\end{proof}

\begin{defn}
We define:
\begin{enumerate}

\item The class of $\sf CS4$ frames to be the class of backward confluent bi-intuitionistic frames.

\item The class of $\sf IS4$ frames to be the class of forward confluent, infallible $\sf CS4$ frames.

\item The class of $\sf S4I$ frames to be the class of forward and downward confluent, infallible bi-intuitionistic frames.

\item The class of $\sf GS4$ frames to be the class of locally linear $\sf IS4$ frames.

\end{enumerate}
\end{defn}

The logics $\sf S4I$ and $\sf GS4$ seem to be new, but there is sound motivation for both of them.
Modal logics based on G\"odel logic, including $\sf GS5$, have already been studied \cite{CaicedoMRR13}, so it is natural to consider $\sf GS4$.
Moreover, we are basing this logic on a sub-class of Simpson's models for $\sf IS4$ \cite{Simpson94}.

Regarding $\sf S4I$, it is easy to check that $(W,\peq,\rel)$ is an $\sf S4I$ frame iff $(W,\rel,\peq)$ is an $\sf IS4$ frame, so it is natural to consider that we are ``commuting'' the roles of $\sf S4$ and intuitionistic logic, which we may denote $\sf I$.
In fact, in the context of expanding products of modal logics, $\sf IS4$ frames are similar to ${\sf I} \times^e {\sf S4}$ frames, where $\times^e$ is the `expanding product' as defined in \cite{pml}, and similarly, $\sf S4I$ frames can be regarded as ${\sf S4} \times^e {\sf I} $ frames.
Finally, in view of Lemma \ref{lemClassical}, the $\sf S4I$ conditions allow us to evaluate $\ps$ and $\nec$ classically.

\section{Soundness}\label{secSound}

In this section we establish the soundness of the axiom schemes we consider for various classes of bi-intuitionistic frames.
We begin with those axioms that are valid on the class of {\em all} bi-intuitionistic frames.

\begin{prop}\label{prop:validity0}
The axioms~\ref{ax:k:box},~\ref{ax:k:dia},~\ref{ax:ref:box},~\ref{ax:ref:dia}, and~\ref{ax:trans:dia} are valid and the inference rules~\ref{rl:mp} and~\ref{rl:nec} preserve validity over any class of bi-intuitionistic frames.
\end{prop}

\begin{proof}
The proofs are standard (and see e.g.~\cite{Simpson94}).
We check only~\ref{ax:trans:dia}. Let $\mathcal M = (W,\fallible W,\peq,\sqsubseteq,V)$ be any bi-intuitionistic model.
Suppose that $(\mathcal M,w) \models \ps\ps p$ and let $v\seq w$.
Then, there is $u\sqsupseteq v$ so that $(\mathcal M,u) \models \ps p $.
Since $u\seq u$, there is $u'\sqsupseteq u$ so that $(\mathcal M,u') \models p $.
By transitivity, $u'\sqsupseteq v$, and since $v$ was arbitrary, $(\mathcal M,w) \models \ps p$.
\end{proof}

\begin{prop}\label{prop:validity1}
The axiom~\ref{ax:trans:box} is valid over any frame that is either backward confluent or downward confluent.
\end{prop}
\begin{proof}
The case for backward confluence is known, as such a frame is a $\sf CS4$ frame~\cite{AlechinaMPR01}.
If instead $\mathcal M = (W,\fallible W,\peq,\sqsubseteq,V)$ is downward-confluent, we have by Lemma \ref{lemClassical} that for any $w\in W$, $(\mathcal M,w) \models\nec \varphi$ iff $\forall v\ler w, (\mathcal M,v) \models \varphi$.
Using this characterization, we may reason as in the classical case to conclude that $\mathcal M\models\nec \varphi \to\nec\nec \varphi$.
\end{proof}	
	
\begin{prop}\label{prop:validity2}
	\begin{enumerate*}[itemsep=0pt,label=(\arabic*)]
		\item\label{itOne} ${\sf IS4} \models \ref{ax:fs}$; 
		\item\label{itTwo} ${\sf IS4} \models \ref{ax:dp}$ and ${\sf S4I} \models \ref{ax:dp}$;
		\item\label{itThree} ${\sf IS4} \models \ref{ax:null}$ and ${\sf S4I} \models \ref{ax:null}$; 
		\item\label{itFour} ${\sf S4I} \models \ref{ax:cd}$;
		\item\label{itFive} ${\sf GS4} \models \ref{ax:g}$.
	\end{enumerate*}	
\end{prop}

\begin{proof}
Items \ref{itOne}-\ref{itThree} are proven in \cite{Simpson94}.
For the remaining items, fix a model $\mathcal M=(W,\fallible W,\peq,\sqsubseteq,V)$.
For \ref{itFour}, assume that $\mathcal M$ is forward- and downward-confluent, and that $(\mathcal M,w) \models \nec \left(\varphi \vee \psi\right)$.
If $(\mathcal M,w)  \models \nec \varphi$, there is nothing to prove, so assume that  $(\mathcal M,w) \not \models \nec \varphi$.
		From Lemma \ref{lemClassical} and downward confluence it follows that there is $v \ler w$ such that $(\mathcal M,v) \not \models \varphi$.
		But $(\mathcal M,v)  \models \varphi \vee\psi$, hence $(\mathcal M,v)  \models  \psi$, and from Lemma \ref{lemClassical} but now using forward confluence,  $(\mathcal M,w)  \models \ps \psi$.
	
Item \ref{itFive} is also well known, but we provide a proof.
Assume by contradiction that ${\sf GS4} \not \models \ref{ax:g}$. 
		This means that  $(\mathcal M,w) \not \models \varphi \to \psi$ and  $(\mathcal M,w)\not  \models  \psi \to \varphi$ for some $w\in W$.
		From the former assumption it follows that there is $v\seq w$ such that $(\mathcal M,v)  \models \varphi$ and  $(\mathcal M,v) \not \models \psi$.
		From the latter assumption it follows that there is $v'\seq w$ such that $(\mathcal M,v')  \models \psi$ and  $(\mathcal M,v') \not \models \varphi$.
		Since ${\mathcal M}$ is locally linear, we need to consider two cases: if $v \peq v'$ we get that  $(\mathcal M,v')  \models \varphi$; if $v'\peq v$ we conclude that  $(\mathcal M,v)  \models \psi$. In any case we reach a contradiction.
\end{proof}

\section{Completeness of ${\sf CS4}$}\label{secCompCS4}

In this section we prove that $\sf CS4$ is complete for its class of models.
This claim is already made in \cite{AlechinaMPR01} and the main elements of the proof are sketched, but there does not seem to be a detailed proof in the literature.
Moreover, our proof of the finite model property relies on specific properties of the canonical model, and establishing these properties will amount to the bulk of the proof of completeness.
For these two reasons, we provide a full completeness proof here.

Fix a logic $\Lambda$.
We use the standard Gentzen-style interpretation of defining $\Gamma \vdash \Delta$ if $\vdash \bigwedge \Gamma' \to \bigvee \Delta '$ for some finite $\Gamma'\subseteq \Gamma$, $\Delta'\subseteq\Delta$.
The logic $\Lambda$ will always be clear from context, which is why we do not reflect it in the notation.
When working within a turnstyle, we will follow the usual proof-theoretic conventions of writing $\Gamma,\Delta$ instead of $\Gamma \cup \Delta$ and $\varphi$ instead of $\{\varphi\}$. By $\Lambda \vdash \varphi$ we mean that the formula $\varphi$ is derivable in the logic $\Lambda$.

A set $X$ of formulas of $\lanfull$ is called {\em prime} if it is closed
under derivation in $\Lambda$ ($X \vdash \varphi$ implies $\varphi \in  X$) and such that $(\varphi \vee \psi) \in  X$ implies that either $\varphi \in  X$ or $ \psi \in  X$.
The prime set $X$ is {\em proper} if $ \bot \not \in X $.

A {\em pre-theory} $\Phi$ consists of two sets of formulas denoted by  $(\Phi^+ ; \Phi^{\ps})$.
We say that $\Phi$ is a {\em $\Lambda$-theory} (or {\em theory} if $\Lambda$ is clear from context) if $\Phi^+$ is a prime set and for any nonempty finite set $\Psi \subseteq \Phi^{\ps}$, $\ps \bigvee \Psi \not\in \Phi^+$.
The intuition behind this definition is the following: formulas in $\Phi^+$ are the ones validated by the theory and the ones in $\Phi^{\ps}$ are those formulas $\psi$ such that $\ps\psi$ is falsified directly via $\rel$ (as opposed to being falsified via a $(\peq;\rel)$-accessible world).

\begin{defn}[]\label{def:deduction} 
	Let $\Lambda$ be a logic over $\lanfull$.
	Given a set of formulas $\Xi$, we say that a pre-theory $\Phi$ is {\em $\Xi$-consistent} if for any finite set $\Delta \subseteq \Phi^{\ps}$,
		\begin{displaymath}
		 \Phi^+ \not \vdash   \Xi,  \ps\bigvee \Delta .
		\end{displaymath}
		We adopt the convention that $\ps \bigvee \varnothing := \bot$.
		We say that $\Phi$ is {\em consistent} if it is $\varnothing$-consistent.
		If $\Xi$ is a singleton $\{\psi\}$, we write $\psi$-consistent instead of $\{\psi\}$-consistent.
\end{defn}

Note that if $\Phi$ is a $\Xi$-consistent theory, then $\Phi^+$ is forcibly proper.
	Below, we say that a pre-theory $\Psi$ {\em extends} $\Phi$ if $\Phi^+\subseteq\Psi^+$ and $\Phi^\ps\subseteq\Psi^\ps$.
	
\begin{lem}[Adapted from~\cite{AlechinaMPR01}]\label{lem:lindembaum} Any $\Xi$-consistent pre-theory $(\Phi^+; \Phi^{\ps})$ can be extended to a $\Xi$-consistent theory $(\Phi^+_*; \Phi^{\ps})$ such that $\Phi^+ \subseteq \Phi^+_*$.
\end{lem}	
	
	\begin{proof}
	By a standard application of Zorn's lemma, there exists a maximal (with respect to set inclusion) set of formulas $ \Phi^+ _* \supseteq \Phi^+$ such that $ (\Phi^+_* ; \Phi^{\ps})$ is $\Xi$-consistent.
	We only need to check that $\Phi^+_*$ is prime.
	Suppose that $\varphi \vee \chi \in \Phi^+_*$.
	We cannot have that $ (  \Phi^+_*,\varphi ; \Phi^{\ps})$ and $ (  \Phi^+_*,\chi; \Phi^{\ps})$ are both $\Xi$-inconsistent, since by left disjunction introduction (admissible in intuitionistic logic) we would obtain that $ (  \Phi^+_*, \varphi \vee \chi ; \Phi^{\ps})$ is $\Xi$-inconsistent.
	However, the latter is just $ (  \Phi^+_* ; \Phi^{\ps})$, contrary to assumption.
	Thus one of the two is $\Xi$-consistent; say, $ (  \Phi^+_*, \varphi ; \Phi^{\ps})$.
	But, by maximality of $\Phi^+_*$, we must already have $\varphi\in \Phi^+_*$, as required.
		\end{proof}

The proof of the following saturation lemma is standard~\cite{AlechinaMPR01}.

\begin{lem}[]
Every theory $\Phi = (\Phi^+; \Phi^\ps)$ satisfies the following properties:

\begin{enumerate}
	\item $\Phi^+$ is deductively closed, i.e. if $\Phi^+ \vdash \varphi$ then $\varphi \in \Phi^+$;
	\item if $\varphi\wedge\psi\in \Phi ^ +$, then $\varphi,\psi\in \Phi^+$, \label{cond:tbl:posconj}
	\item if $\varphi\wedge\psi \not \in \Phi ^ +$, then	$\varphi \not \in \Phi ^+$ or $\psi \not \in \Phi^+$, \label{cond:tbl:negconj}
	\item if $\varphi\vee\psi\in \Phi ^ +$, then	$\varphi \in \Phi^+$ or $\psi\in \Phi^+$, \label{cond:tbl:posdisj}
	\item if $\varphi\vee\psi\not\in \Phi ^+$, then  $\varphi , \psi\not \in \Phi^+$, \label{cond:tbl:negdisj}
	\item if $\varphi\to\psi\in \Phi^+$, then $\varphi \not \in \Phi^+$ or $\psi \in\Phi^+$, \label{cond:tbl:implication}
	\item if $\nec\varphi\in \Phi^+$, then $\varphi \in \Phi^+$,
	\item if $\varphi \in \Phi^\ps$, then $\ps \varphi \not \in \Phi^+$, \label{cond:subset:diam}

\end{enumerate}
\end{lem}

If $\Phi$ is a theory then $\Phi^\nec \eqdef \{ \varphi \in \lanfull:\nec\varphi\in \Phi^+ \}$.

\begin{defn}[Adapted from~\cite{AlechinaMPR01}]
We define the canonical model for ${\sf CS4}$ as $\mathcal M_c^{{\sf CS4}}=(W_c,\fallible W_c,\peq_c,\rel_c, V_c)$, where~
\begin{itemize}

	\item $W_c$ is the set of all $\sf CS4$-theories;
	\item $\fallible W_c = \{ (\lanfull;\varnothing) \}$;
	\item ${\peq_c} \subseteq W_c\times W_c$ is defined by $\Phi \peq_c \Psi$ iff $\Phi^+ \subseteq \Psi^+$;
	\item ${\rel_c} \subseteq W_c\times W_c$ defined as $\Phi \rel_c \Psi$ iff $\Phi^\nec\subseteq \Psi^+$ and $\Phi^\ps \subseteq \Psi^{\ps}$;
	\item $V_c$ is defined by $V_c(p)  = \lbrace \Phi \mid p \in \Phi^+  \rbrace$.

\end{itemize}
\end{defn}

\begin{remark}
The notation ${W_c}$, $\peq_c$, etc.~will also be used for the canonical models of logics distinct from $\sf CS4$.
The meaning of the notation will be made clear at the beginning of each section and will remain constant throughout.
In this section, this notation will always refer to the components of the structure $\mathcal M_c^{{\sf CS4}}$.
\end{remark}

\begin{lem}\label{lemCS4IsModel}
$\mathcal M_c^{{\sf CS4}}$ is a bi-intuitionistic model.
\end{lem}

\begin{proof}
It is easy to see that $\peq_c$ is a preorder given that $\subseteq$ is itself a preorder, and the monotonicity conditions for $V_c$ and $\fallible W_c$ are easily verified. We focus on showing that $\rel_c$ is a preorder.

Let us take $\Phi \in W_c$ and $\nec \varphi \in \Phi^+$. 
	Due to the Axiom~\ref{ax:ref:box} and~\ref{rl:mp}, $\varphi \in \Phi^+$. 
	Therefore, $ \Phi^\nec \subseteq \Phi^+$. 
Trivially, $\Phi^{\ps} \subseteq \Phi^{\ps}$, so in conclusion, $\rel_c$ is reflexive. 	

	
To see that $\rel_c$ is transitive, let us consider $\Phi, \Psi, \Omega \in W_c$ such that $\Phi \rel_c \Psi$ and $\Psi \rel_c \Omega$. 
	Let us take $\nec \varphi \in \Phi^+$. 
	From Axiom~\ref{ax:trans:box} and~\ref{rl:mp}, $\nec \nec \varphi \in \Phi^+$.
	From $\Phi \rel_c \Psi$ and $\Psi \rel_c \Omega$ we get $\nec \varphi\in \Psi^+$ and $\varphi\in \Omega^+$.
	Therefore $ \Phi^\nec \subseteq \Omega^+$.

	Let us take now $\varphi \in \Phi^{\ps}$. 
	Since $\Phi \rel_c \Psi$ then $\varphi \in \Psi^{\ps}$.
	Since $\Psi \rel_c \Omega$ then $\varphi \in \Omega^{\ps}$.
	This means that $\Phi^{\ps} \subseteq \Omega^{\ps}$.
	Consequently, $\Phi \rel_c\Omega$, so $\rel_c$ is transitive.
\end{proof}


\begin{prop}\label{prop:backward}
$\mathcal M_c^{{\sf CS4}}$ is backward confluent.
In particular, if $\Phi \rel_c \Psi \peq_c \Omega$, then $\Upsilon := (\Phi^+;\varnothing) $ satisfies $\Phi \peq_c \Upsilon \rel_c \Omega$.
\end{prop}

\begin{proof}
Take $\Phi, \Psi$ and $\Omega$ in $W_c$ such that $\Phi \rel_c \Psi \peq_c \Omega$ and let us define $\Upsilon = \left( \Phi^+ , \varnothing \right)$.
$\Phi^+$ is already prime, so $\Upsilon \in W_c$. By definition $\Phi^+ \subseteq \Upsilon^+$, thus $\Phi \peq_c \Upsilon$, and $\Upsilon^\ps \subseteq \Omega^{\ps}$. 
Take $\nec \varphi \in \Upsilon^+$. By definition $\nec \varphi \in \Phi^+$. Since $\Phi \rel_c \Psi \peq_c \Omega$ then $\varphi \in \Omega^+$.
Since $\nec \varphi$ was chosen arbitrarly, $ \Upsilon^\nec \subseteq \Omega^+$, so $\Phi \peq_c \Upsilon \rel_c \Omega$.
\end{proof}

\begin{lem}\label{lemDiamCS4}for all $\Gamma \in W_c$ and for all $\nec \varphi$ and $\ps \varphi$ in $\lanfull$, the following items hold.
	\begin{itemize}
		\item $\ps\varphi\in \Gamma^+$ if and only if for all $\Psi\seq_c \Gamma$ there is $\Delta\in W_c$ such that $\Psi \rel_c \Delta$ and $ \varphi\in \Delta^+$.
		\item $\nec\varphi\in \Gamma^+$ if and only if for all $\Psi$ and $\Delta$ such that $\Gamma \peq_c \Psi\rel_c \Delta$, $ \varphi\in \Delta^+$.
	\end{itemize}	
\end{lem}

\begin{proof}
	We start with the first item. From left to right, assume that $\ps\varphi \in \Gamma^+$, and let $\Psi\seq_c\Gamma$. We claim that there is $\Delta\ler_c \Psi$ such that $\varphi \in \Delta^+$. 
	Let us take $\Upsilon = (  \Psi^\nec, \lbrace\varphi \rbrace;  \Psi^\ps)$. 
	We show that $\Upsilon$ is consistent.
	If not, let $\chi_1,\ldots,\chi_n \in  \Psi^\nec$ and $\psi_1,\ldots,\psi_n \subseteq \Psi^\ps$ be such that, for $\chi :=\bigwedge_i \chi_i$ and $\psi :=\bigvee_i \psi_i$, ${\sf CS4} \vdash \chi \wedge \varphi \to \ps \psi$, so that ${\sf CS4} \vdash \chi \to\left( \varphi \to \ps\psi\right)$.
	By~\ref{rl:nec} and~\ref{ax:k:box} it follows that ${\sf CS4} \vdash \nec \chi \to \nec \left(\varphi \to \ps\psi\right)$.
	By~\ref{rl:mp}, $\nec \left(\varphi \to \ps\psi\right) \in \Psi^+$ and, by Axiom~\ref{ax:k:dia}, $\ps\varphi \to \ps\ps\psi \in \Psi^+$.
	Since $\Gamma \peq_c \Psi$ and $\ps \varphi \in \Gamma^+$ then $\ps \varphi \in \Psi^+$.
	By~\ref{rl:mp}, $\ps \ps \psi \in \Psi^+$.
	By Axiom~\ref{ax:trans:dia}, $\ps \psi \in \Psi^+$, which contradicts the consistency of $\Psi$.
	We conclude that $\Upsilon$ is consistent.
	By Lemma~\ref{lem:lindembaum}, $\Upsilon$ can be extended to a theory $\Delta= (\Delta^+;\Delta^\ps)$, such that $\Delta \in W_c$, $\nec \Psi^+ \subseteq \Delta^+$ and $\Psi^\ps \subseteq \Delta^{\ps}$, therefore $\Psi \rel_c \Delta$. 
	Moreover $\varphi \in \Delta^+$, as needed.
	
	Conversely, let us assume that $\ps \varphi \not\in \Gamma^+$ and let us define $\Psi = (\Gamma^+;\lbrace \varphi \rbrace)$.
 	It is easy to see that $\Psi$ is consistent, and since $\Gamma^+$ is prime, $\Psi \in W_c$. 
	Moreover, $\Gamma \peq_c \Psi$.
	We claim that for all $\Delta \in W_c$, if $\Psi \rel_c \Delta$ then 
	$\varphi \not \in \Delta^+$.
	To prove it, let us take any $\Delta \in W_c$ satisfying $\Psi \rel_c \Delta$. 
	By definition, $\Psi^\ps \subseteq \Delta^{\ps}$, so $\varphi \in \Delta^{\ps}$.
	By definition, $\ps \varphi \not \in \Delta^+$, so $\varphi \not \in \Delta^+$ because of Axiom~\ref{ax:ref:dia}, as needed.

	Let us consider now the case of $\nec \varphi$. From left to right, let $\Psi$ and $\Delta$ be such that $\nec\varphi \in \Gamma^+$ and $\Gamma \peq_c \Psi \rel_c \Delta$. We claim that $\varphi \in \Delta^+$.
	Since $\Gamma \peq_c \Psi \rel_c \Delta$, $\Gamma^+ \subseteq \Psi^+$ and $ \Psi^\nec \subseteq \Delta^+$. Since $\nec \varphi \in \Gamma^+$ then $\varphi \in \Delta^+$, as needed.

	Conversely, let us assume that $\nec \varphi \not \in \Gamma^+$ and let us define $\Psi  = (\Gamma^+;\varnothing)$. Obviously, $\Psi$ is a theory and it satisfies $\Gamma \peq_c \Psi$.
	Let us take $\Upsilon = ( \Psi^\nec , \varnothing )$. $\Upsilon$ is $\varphi$-consistent, since otherwise \ref{rl:nec},  \ref{ax:k:box}, and \ref{rl:mp} would yield that $\nec\Psi^\nec \vdash \nec\varphi$, contradicting $\nec \varphi \not \in \Gamma^+$.
	In view of Lemma~\ref{lem:lindembaum}, $\Upsilon$ can be extended to a $\varphi$-consistent theory $\Delta = (\Delta^+;\varnothing) \in W_c$ such that $\Upsilon^+ \subseteq \Delta^+$. By definition of $\Upsilon$, $\Psi \rel_c \Delta$ and $\varphi \not \in \Delta^+$, as needed.  
\end{proof}

\begin{lem}[Truth Lemma]\label{lem:truth-lemma}
	For any theory $\Phi \in W_c$ and $\varphi \in \lanfull$, 
\[\varphi \in \Phi^+ \Leftrightarrow ( \mathcal M_c^{{\sf CS4}}, \Phi) \models \varphi.\]
\end{lem}
\begin{proof}
	By induction on the complexity of $\varphi$.
	The case of propositional variables is proven by the definition of $V_c$.
	The case of $\wedge$ and $\vee$ are proved by induction.
	We consider the $\to$ connective next.

Assume that $\varphi \to \psi \in \Phi^+$ and let $\Psi \seq_c \Phi$ be such that $(\mathcal M_c^{{\sf CS4}}, \Psi )\models \varphi$.
By the induction hypothesis, $\varphi\in \Psi^+$.
		Since $\Phi \peq_c \Psi$ then $\varphi \to \psi \in \Psi^+$.
		By~\ref{rl:mp}, $\psi \in \Psi^+$, and by the induction hypothesis, $(\mathcal M_c^{{\sf CS4}}, \Psi ) \models  \psi$.
Since $\Psi$ was arbitrary, $(\mathcal M_c^{{\sf CS4}}, \Phi ) \models \varphi\to \psi$.

If $\varphi \to \psi \not \in \Phi^+$, let us take $\Upsilon = (\Phi^+, \lbrace \varphi \rbrace; \varnothing)$.
		We claim that $\Upsilon$ is $\psi$-consistent. 
		If not, by definition of $\vdash$, there exists $\chi \in \Phi^+$ such that ${\sf CS4} \vdash \chi \wedge  \varphi \to \psi$. 
		It follows that  ${\sf CS4} \vdash \chi \to\left(  \varphi \to \psi\right)$.
		Since $\chi \in \Phi^+$ then $\varphi \to \psi \in \Phi^+$: a contradiction.
		Hence $\Upsilon$ is $\psi$-consistent, so that by Lemma~\ref{lem:lindembaum}, $\Upsilon$ can be extended to a maximal $\psi$-consistent $\Psi=(\Psi^+;\varnothing)$ such that $\Upsilon^+ \subseteq \Phi^+$.
		By induction on $\varphi$ and $\psi$, $(\mathcal M_c^{{\sf CS4}}, \Psi) \models \varphi$ and  $(\mathcal M_c^{{\sf CS4}}, \Psi) \not \models \psi$.
		By how $\Upsilon$ is defined, $\Phi^+ \subseteq \Psi^+$, so $\Phi \peq_c \Psi$.
		Therefore, $(\mathcal M_c^{{\sf CS4}}, \Phi) \not \models \varphi \to \psi$
	
	The case of the $\nec \varphi$ formulas is proved next.
If $\nec \varphi \in \Phi^+$, by Lemma~\ref{lemDiamCS4}, for all $\Gamma \peq_c \Psi \rel_c \Omega$, $\varphi \in \Omega^+$.
		By induction $(\mathcal M_c^{{\sf CS4}}, \Omega) \models \varphi$.
		As a consequence, $(\mathcal M_c^{{\sf CS4}}, \Phi) \models \nec \varphi$.

If $\nec \varphi \not \in \Phi^+$, thanks to Lemma~\ref{lemDiamCS4}, there exist $\Phi \peq_c \Psi \rel_c \Omega$ such that $\varphi \not \in \Omega^+$. By induction $(\mathcal M_c^{{\sf CS4}}, \Omega) \not \models \varphi$. As a consequence, $(\mathcal M_c^{{\sf CS4}}, \Phi) \not \models \nec \varphi$.		

	We finish by considering the case of $\ps \varphi$. If $\ps\varphi \in \Phi^+$, by Lemma~\ref{lemDiamCS4} for all $\Phi \peq_c \Psi$ there exists $\Psi \rel_c \Omega$ such that $\varphi \in \Omega^+$.
		By induction $( \mathcal M_c^{{\sf CS4}}, \Omega ) \models \varphi$.
		Consequently, $  ( \mathcal M_c^{{\sf CS4}}, \Phi ) \models \ps\varphi$.
		

If $\ps \varphi \not\in \Phi^+$, by Lemma~\ref{lemDiamCS4}, there exists $\Phi \peq_c \Psi$ such that for all $\Psi \rel_c \Omega$, $\varphi \not \in \Omega^+$. By induction, if $\Psi \rel_c \Omega$ then $ ( \mathcal M_c^{{\sf CS4}}, \Omega ) \not \models \varphi$. From the semantics, $( \mathcal M_c^{{\sf CS4}}, \Phi ) \not \models \ps \varphi$.
\end{proof}

It follows from the above considerations that $\sf CS4$ is complete for its class of models, and in particular, any formula that is not derivable is falsifiable on $\mathcal M^{\sf CS4}$.
We defer the formal statement to Section \ref{secCompS4I}, after we have constructed the canonical models for the other logics we consider.

\section{Completeness of ${\sf GS4}$}\label{secCompGS4}

In this section we show that $\sf GS4$ is complete by techniques analogous to those used for the completeness of $\sf CS4$.
In particular, we continue working with theories.
Say that a theory $\Phi$ is {\em precise} if for each formula $\psi$, $\ps \psi \in \Phi^+ \cup \Phi^{\ps}$.
In the case of a precise theory, any false instance of $\ps\psi$ is already falsified via $\rel$, according to the semantics for $\ps$ on $\sf GS4$ frames, as given by Lemma \ref{lemClassical}.\footnote{In fact, the component $\Phi^{\ps}$ is not required for treating $\sf GS4$, but it is convenient for the sake of keeping some later proofs uniform.}

We need the following strengthening of Lemma \ref{lem:lindembaum} which yields {\em precise} theories.

\begin{lem}\label{lem:lindembaumPrecise} Any $\Xi$-consistent pre-theory $(\Phi^+; \Phi^{\ps})$ can be extended to a precise, $\Xi$-consistent theory $(\Phi^+_*; \Phi^{\ps})$ such that $\Phi^+ \subseteq \Phi^+_*$.
\end{lem}

\begin{proof}[Proof sketch]
First, extend $\Phi$ to $(\Phi^+_*,\Phi^\ps)$ as in the proof of Lemma \ref{lem:lindembaum}.
Then, define $\Phi^{\ps}_* = \{\varphi: \ps\varphi\not\in \Phi^+_*\} $.
It is not hard to check that $\Phi_*$ is a precise theory extending $\Phi$.
\end{proof}

In fact Lemma \ref{lem:lindembaumPrecise} already holds over $\sf CS4$, but precise theories were not needed then.
We define the canonical model for ${\sf GS4}$ as $\mathcal M_c^{{\sf GS4}}=(W_c,\peq_c,\rel_c, V_c)$, where
\begin{enumerate*}[label=\alph*)]
	\item $W_c$ is the set of precise, consistent\footnote{Following the literature on $\sf IS4$, $\sf GS4$ accepts the axiom $\ref{ax:null}:=\neg\ps\bot$, which is valid on the class of infallible models, i.e., $\bot$ should be false on all worlds; this is why we only accept consistent theories.
	Note, however, that our proofs do not rely on this axiom: a `fallible' version of $\sf GS4$ could be considered, as well as an `infallible' version of $\sf CS4$, and our proofs and results would go through mostly unchanged.} $\sf GS4$-theories $\Psi = (\Psi^+;\Psi^{\ps})$;
	\item $\fallible W_c =\varnothing$;
	\item $\peq_c$, $\rel_c$ and $V_c$ are defined as for $ M_c^{{\sf CS4}}$.	
\end{enumerate*}
We leave to the reader to verify that $\mathcal M_c^{{\sf GS4}} $ is an infallible bi-intuitionistic model, and instead focus on showing that it satisfies the $\sf GS4$ frame conditions of local linearity and forward and backward confluence.


\begin{lem} $\mathcal M_c^{\sf GS4}$ is locally linear.
\end{lem}
\begin{proof} 
	Assume toward a contradiction that $\peq_c$ is not locally linear. 
	Let the theories $\Phi$, $\Psi$ and $\Omega$ be such that $\Phi \peq_c \Psi$, $\Phi \peq_c \Omega$ but $\Psi \not \peq_c \Omega$ and $\Omega \not \peq_c \Psi$. 
	From the definition of $\peq_c$ we get that $\Psi^+ \not \subseteq \Omega^+$ and $\Omega^+ \not \subseteq \Psi^+$.
	Therefore, there exists two formulas $\varphi$ and $\psi$ such that $\varphi \in \Psi^+$, $\varphi \not\in \Omega^{+}$, $\psi \in \Omega^+$ and $\psi \not\in \Psi^{+}$.
	Therefore $\varphi \to \psi \not \in \Psi^+ \supseteq \Phi^+$ and $\psi \to \varphi \not \in \Omega^+ \supseteq \Phi^+$. 
	Consequently, $(\varphi \to \psi) \vee (\psi \to \varphi) \not \in \Phi^{+}$: a contradiction.
\end{proof}

\begin{lem}
$\mathcal M_c^{\sf GS4}$ is forward and backward confluent.
\end{lem}

\begin{proof}
We first check that it satisfies forward confluence.
Let $\Phi$, $\Psi$ and $\Theta$ in $W_c$ be such that $\Phi \peq_c \Psi$ and $\Phi \rel_c \Theta$. 
We claim that $( \Psi^\nec, \Theta^+;\Psi^\ps)$ is consistent.
If not, there exist $  \varphi \in \Psi^\nec$, $\chi \in \Theta^+$ and $  \psi  \in \Psi^\ps$ such that
${\sf GS4} \vdash \varphi \wedge \chi \rightarrow \ps\psi$ (note that we can take single formulas since $ \Psi^\nec,\Theta^+$ are closed under conjunction and, in view of \ref{ax:dp}, $\Psi^\ps$ is closed under disjunction).
Since $\chi \in \Theta^+$ then $\varphi \rightarrow \ps\psi \in \Theta^+$. 
Thanks to Axiom~\ref{ax:ref:dia}, $\ps \left(\varphi \rightarrow \ps\psi\right) \in \Theta^+$.
Therefore, $\varphi \rightarrow \ps\psi \not \in  \Theta^{\ps}$, so $\varphi \rightarrow \ps\psi \not \in  \Phi^{\ps}$.
This means that $\ps \left(\varphi \rightarrow \ps\psi\right) \in  \Phi^+$.
Using the derivable formula~\ref{der:ax1} we get $\nec \varphi \rightarrow \ps\ps\psi \in \Phi^+$. 
Since $\Phi^+ \subseteq \Psi^+$ it follows that $\nec \varphi \rightarrow \ps\ps\psi \in \Psi^+$, so $\ps\ps \psi \in \Psi^+$. By Axiom~\ref{ax:trans:dia}: a contradiction. 

In view of Lemma~\ref{lem:lindembaumPrecise}, $( \Psi^\nec, \Theta^+;\Psi^\ps)$ can be extended to a precise, consistent theory $\Upsilon$.
%
%
%
It is easy to check by our choice of $\Upsilon$ that $\Psi \rel_c \Upsilon$ and $\Theta \peq_c \Upsilon$.
%
%
%
%
%

Next we check that $\mathcal M^{\sf GS4}$ is backward confluent.
Suppose that $\Phi\rel_c\Psi\peq_c \Theta$, and let $\Xi := \{\nec\xi \in\lanfull : \xi\notin \Theta^+\}$ and $\Delta = \{\ps \delta \in\lanfull : \ps\delta \in \Theta^+\}$.
We claim that $(\Phi^+,\Delta; \varnothing )$ is $\Xi$-consistent.
If not, there are $\varphi \in\Phi^+$, $\ps\delta_1,\ldots,\ps\delta_n \in \Delta$, and $\nec\xi\in \Xi$ such that $ \varphi ,\{\ps \delta_i\}_{i=1}^n \vdash \nec \xi $, which using \ref{ax:trans:dia} and \ref{ax:trans:box} yields $ \varphi \vdash \bigwedge_{i=1}^n \ps \ps \delta_i \to \nec\nec \xi $.
By repeated applications of \ref{ax:fs}, we obtain $\varphi \vdash \nec\left (\bigwedge_{i=1}^n \ps \delta_i \to \nec \xi \right )$.
Since $\varphi\in \Phi^+$, $\nec\left (\bigwedge_{i=1}^n \ps \delta_i \to \nec \xi \right ) \in\Phi^+$.
Since $\Phi\rel_c\Psi$, $\bigwedge_{i=1}^n \ps \delta_i \to \nec \xi \in \Psi^+$, and since $\Psi\peq_c\Theta$ and each $\ps\delta_i\in \Theta^+$, $\nec \xi\in \Theta^+$, which by \ref{ax:ref:box} implies that $\xi\in \Theta^+$, contradicting our choice of $\Xi$.

Hence $(\Phi^+,\Delta; \varnothing )$ is $\Xi$-consistent.
Then, any prime, $\Xi$-consistent extension $\Upsilon$ satisfies $\Phi\peq_c\Upsilon\rel_c\Theta$: that $\Phi\peq_c\Upsilon$ follows from $\Phi^+\subseteq \Upsilon^+$.
We check that $\Upsilon\rel_c\Theta$.
We have that $\Upsilon^\nec\subseteq \Theta^+$, since if $\xi\notin \Theta^+$ it follows by the definition of $\Xi$ and the fact that $\Upsilon$ is $\Xi$-consistent that $\xi\notin\Upsilon^\nec$, while if $\delta\notin\Theta^\ps$ it follows that $\ps\delta\in \Theta^+$, hence $\ps\delta\in \Delta \subseteq \Upsilon^+$ and $\delta\notin\Upsilon^\ps$.
It follows that $\Upsilon^\ps\subseteq \Theta^\ps$, so $\Upsilon\rel_c\Theta$.
\end{proof}

\begin{lem}[Truth Lemma]\label{lem:truth-lemma:GS4}
	For any theory $\Phi \in W_c$ and $\varphi \in \lanfull$, 
\[\varphi \in \Phi^+ \Leftrightarrow ( \mathcal M_c^{{\sf GS4}}, \Phi ) \models \varphi.\]
\end{lem}

\begin{proof} We consider the modalities, as other connectives are treated as in the case for $\sf CS4$.
The left-to-right implication for the case $\nec\varphi$ is also treated as in the $\sf CS4$ case.

For the other direction, we work by contrapositive.
If $\nec \varphi \not \in \Phi^+$, using Zorn's lemma, let $\Psi\seq_c \Phi$ be $\peq_c$-maximal so that $\nec\varphi\not\in \Psi^+$.
It is readily checked using maximality that $\Psi^+$ is prime.

We claim that
\begin{equation}\label{eqDiamImp}
\psi\in \Psi^\ps \Rightarrow \nec(\psi\to \varphi) \in \Psi^+.
\end{equation}
By maximality of $\Psi$, we have that $\Psi^+,\ps\psi \vdash \nec \varphi $, so $\Psi^+\vdash \ps\psi\to\nec\varphi$.
By \ref{ax:fs}, $\Psi^+\vdash \nec(\psi\to\varphi)$, as needed.

We claim that $(\Psi^\nec;\Psi^\ps)$ is $\varphi$-consistent.
Noting by \ref{ax:dp} that $\Psi^\ps$ is closed under disjunction, if not, we would have $\chi\in \Psi^\nec$ and $\theta \in \Psi^\ps$ such that $ \chi \vdash  \varphi\vee \ps\theta$, and reasoning as before, $ \nec\chi  \vdash \nec( \varphi\vee \ps\theta)$.
It follows that  $ \nec\chi ,\nec(\ps\theta\to \varphi) \vdash \nec( \varphi\vee \varphi)$, i.e.~$ \nec\chi ,\nec(\ps\theta\to \varphi) \vdash \nec \varphi $.
From \ref{ax:trans:dia} we see that $\ps\ps\theta \in \Psi^+$ implies that $ \ps\theta\in \Psi^+$, which by contrapositive becomes $\theta\in \Psi^\ps $ implies $\ps\theta\in\Psi^\ps$.
Thus we may use \eqref{eqDiamImp} to conclude that $\nec(\ps\theta\to\varphi) \in \Psi^+$, so $\Psi^+ \vdash \nec\varphi$, a contradiction.

Hence there is a precise, $\varphi$-consistent theory $\Upsilon$ extending $(\Psi^\nec;\Psi^\ps)$.
The induction hypothesis yields $(\mathcal M^{\sf GS4},\Upsilon)\not\models \varphi$, and it is readily verified that $\Phi\peq_c\Psi\rel_c \Upsilon $, as needed.

For the case of $\ps \varphi$, if $\ps \varphi \in \Phi^+$, then let us define $\Theta =   \varphi , \Phi^\nec$, and let us assume toward a contradiction that $\Theta \vdash \lbrace \psi : \ps \psi \not \in \Phi^+ \rbrace$. This means that there exists $\nec \chi \in \Phi^+$, $\ps \psi \not \in \Phi^+$ such that ${\sf GS4} \vdash \chi \rightarrow\left( \varphi \rightarrow \psi\right)$. From~\ref{rl:nec},~\ref{ax:k:box},~\ref{ax:k:dia}, and~\ref{rl:mp} we get $\ps \psi \in \Phi^+$: a contradiction. Therefore $(\Theta,\varnothing)$ can be extended to a precise theory $\Upsilon $ such that $\Upsilon^+ \not \vdash \lbrace \psi : \ps \psi \not \in \Phi^+ \rbrace$.
		It follows that $\Phi\rel_c\Upsilon$. Moreover, $\varphi \in \Upsilon^+$. By induction hypothesis, $ ( \mathcal M_c^{{\sf GS4}}, \Upsilon )   \models \varphi$, so $( \mathcal M_c^{{\sf GS4}}, \Phi )  \models \ps\varphi$.

Now assume that $ ( \mathcal M_c^{{\sf GS4}}, \Phi ) \models \ps\varphi$, so $( \mathcal M_c^{{\sf GS4}}, \Upsilon) \models \varphi$ for some $\Phi \rel_c \Upsilon$. By induction, $\varphi \in \Upsilon^+$, hence by \ref{ax:ref:dia} $\ps\varphi \in \Upsilon^+$, and thus $\varphi\not\in \Upsilon^\ps$.
It follows that $ \varphi\not\in \Phi^\ps$, hence $\ps\varphi\in \Phi^+$, as required.
\end{proof}

As before, completeness of $\sf GS4$ readily follows, but we defer the general completeness statement to the end of the following section.

\section{Completeness of ${\sf S4I}$}\label{secCompS4I}

Next we prove that $\sf S4I$ is complete.
We follow the same general pattern as for $\sf CS4$ and $\sf GS4$.
We define the canonical model for ${\sf S4I}$ as $\mathcal M_c^{{\sf S4I}}=(W_c,\peq_c,\rel_c, V_c)$, defined analogously to $\mathcal M_c^{{\sf S4I}}$, except that $W_c$ is now the set of proper, consistent $\sf S4I$-theories.
As before, we leave to the reader to check that $\mathcal M_c^{{\sf S4I}}$ is a bi-intuitionistic model, and focus on those properties of the model that are particular to $\sf S4I$.

\begin{lem} $\mathcal M_c^{\sf S4I}$ is forward and downward confluent.\end{lem}

\begin{proof}
Forward confluence is checked as in the case of $\sf GS4$, so we focus on downward confluence.
Let $\Phi$, $\Psi$ and $\Theta$ in $W_c$ be such that $\Phi \peq_c \Psi\rel_c \Theta$.
Let $\Xi = \lanfull\setminus \Theta^+$.
We claim that $(\Phi^\nec;\Phi^\ps)$ is $\Xi$-consistent.
If not, there exist $\nec \varphi \in \Phi^+$, $\ps \psi \not \in \Phi^+$ and $\chi \not \in \Theta^+$ such that ${\sf S4I} \vdash \varphi \rightarrow   \chi \vee \ps\psi$.
Due to~\ref{rl:nec}, Axiom~\ref{ax:k:box} and~\ref{rl:mp} we get that $\nec \left(\chi \vee \ps\psi\right) \in \Phi^+$.
By~\ref{ax:cd} it follows that $\nec \chi \vee \ps \psi \in \Phi^+$.
Since $\chi \not \in \Theta^{+}$ then $\nec \chi \not \in \Phi^{+}$.
Consequently $\ps \psi \in \Phi^{+}$: a contradiction. 
Thanks to Lemma~\ref{lem:lindembaumPrecise}, $(\Phi^\nec;\Phi^\ps)$ can be extended to a $\Xi$-consistent, precise theory $\Upsilon$.
It then readily follows that $\Phi\rel_c\Upsilon\peq_c\Theta$, as required.
\end{proof}

\begin{lem}[Truth Lemma]\label{lem:truth-lemma:S4I}
	For all theories $\Phi \in W_c$ and $\phi \in \lanfull$, 
	\[\varphi \in \Phi^+ \Leftrightarrow ( \mathcal M_c^{{\sf S4I}}, \Phi ) \models \varphi\]
\end{lem}

\begin{proof}
We consider only the modalities.
For the case of a formula $\nec\varphi$, the left-to-right direction proceeds as usual.		
For the other direction, reason by contrapositive.
If $\nec \varphi \not \in \Phi^+$,  we claim that $ (\Phi^\nec;\Phi^\ps)$ is $\varphi$-consistent. 
		If not, there exists $\nec \chi \in \Phi^+$, $\ps \psi \not \in \Phi^+$ such that ${\sf S4I} \vdash \chi \rightarrow \varphi \vee \ps \psi$.
		By~\ref{rl:nec},~\ref{ax:k:box} and~\ref{rl:mp} we get $\nec ( \varphi \vee \ps \chi)$. By Axiom~\ref{ax:cd}, $\nec \varphi \vee \ps \psi \in \Phi^+$, a contradiction since $\Phi^+$ is prime and $\nec\varphi,\ps\psi\not\in \Phi^+$.
		Therefore, $( \Phi^\nec;\Phi^\ps) $ can be extended to a precise, $\varphi$-consistent theory $\Upsilon $.
		Clearly, $\Phi\rel_c\Upsilon$.
Since $\Upsilon$ is $\varphi$-consistent, $\varphi\not\in \Upsilon^+$.
		By the induction hypothesis, $ ( \mathcal M_c^{{\sf S4I}}, \Upsilon ) \not \models \varphi$. 
		Therefore, $ ( \mathcal M_c^{{\sf S4I}}, \Phi ) \not \models \nec \varphi$.

	 For the case of $\ps \varphi$, the right-to-left direction is standard, so we focus on the other.
	 If $\ps \varphi \in \Phi^+$, we claim that $(\Phi^\nec,\varphi; \Phi^\ps)$ is consistent.
If not, there exists $\nec \chi \in \Phi^+$, $\ps \psi \not \in \Phi^+$ such that ${\sf S4I} \vdash \chi \rightarrow\left( \varphi \rightarrow \ps \psi\right)$.
From~\ref{rl:nec},~\ref{ax:k:box}, and~\ref{rl:mp}, we get $\ps \psi \in \Phi^+$: a contradiction. Therefore $(\Phi^\nec,\varphi; \Phi^\ps)$ can be extended to a precise, consistent theory $\Upsilon$.
		It follows from our construction that $\Phi\rel_c\Upsilon$.
		Moreover, $\varphi \in \Upsilon^+$. By the induction hypothesis, $ ( \mathcal M_c^{{\sf S4I}}, \Upsilon )  \models \varphi$, so $( \mathcal M_c^{{\sf S4I}}, \Phi )  \models \ps\varphi$.
\end{proof}

Let us summarize and state our completeness results.

\begin{thm}\label{thmComp}
For $\Lambda\in \{{\sf CS4},{\sf GS4},{\sf S4I}\}$, $\Lambda$ is sound and complete for the class of $\Lambda$-models.
In particular, $\mathcal M^{\Lambda}_c$ is a $\Lambda$ model, and if ${\Lambda} \not\vdash\varphi$, then $\mathcal M^{\Lambda}_c \not\models \varphi$.
\end{thm}

\begin{proof}
For each logic $\Lambda$, $\mathcal M^{\Lambda}_c$ is a bi-intuitionistic model: this is stated for $\sf CS4$ in Lemma \ref{lemCS4IsModel}, and checked analogously for the other logics.
In each case, we have shown in addition that $\mathcal M^{\Lambda}_c$ satisfies the required frame conditions.
If ${\Lambda} \not\vdash\varphi$, then the pre-theory $\Upsilon=(\varnothing;\varnothing)$ is $\varphi$-consistent, hence can be extended to a $\varphi$-consistent theory $\Phi\in W_c$.
By Lemma \ref{lem:truth-lemma}, $(\mathcal M^{\Lambda}_c,\Phi)\not\models\varphi$.
\end{proof}

\begin{remark}
Theorem \ref{thmComp} also holds for $\sf IS4$; this is proven in \cite{servi1984axiomatizations,Simpson94}, but also follows from the development in Section \ref{secCompGS4}, omitting the proof of local linearity.
\end{remark}

\section{$\Sigma$-bisimulations}\label{secSBisim}

The remainder of the article will be devoted to establishing the finite model property for $\sf CS4$, $\sf GS4$, and $\sf S4I$.
In this section we develop the theory of {\em $\Sigma$-bisimulations,} one of the key components in our proof.

\begin{defn}
	Given a set of formulas $\Sigma$ and a model $\mathcal M =(W,\fallible W,\peq,\rel,V) $, we define the {\em $\Sigma$-label} of $w\in W$ to be a pair $\ell(w) = (\ell^+(w);\ell^\ps (w))$, where
	\begin{align*}
	\ell^+(w) &:= \{ \varphi\in \Sigma : (\mathcal M,w)\models \varphi \},\\
	\ell^\ps(w) &:= \{ \varphi\in \Sigma : \forall v\ler w \ (\mathcal M,v) \not \models \varphi \}.
	\end{align*}
	A {\em $\Sigma$-bisimulation} on $\mathcal M$ is a forward and backward confluent relation $Z \subseteq W  \times W $ so that if $w\mathrel Z v$ then $\ell(w) = \ell(v)$.
	We denote the greatest $\Sigma$-bisimulation by $\sim_\Sigma$.
\end{defn}

	Notice that an arbitrary union of $\Sigma$-bisimulations is a $\Sigma$-bisimulation. Hence, $\sim_{\Sigma}$ is well-defined.
	Obviously, $\sim_{\Sigma}$ is an equivalence relation.
For each canonical model $\mathcal M^\Lambda_c$, we note that $\ell(\Phi)= \Phi\upharpoonright \Sigma$, where the latter is defined by $\Phi\upharpoonright \Sigma:=(\Phi^+\cap \Sigma,\Phi^\ps \cap \Sigma)$.

\begin{defn}
	Given a model $\mathcal M=(W,\fallible W,\peq,\rel,V)$ and an equivalence relation ${\sim}\subseteq W\times W$, we denote the equivalence class of $w\in W$ under $\sim$ by $[w]$.
	We then define the quotient $\nicefrac{\mathcal M}\sim = (\nicefrac W\sim,\nicefrac{\fallible W}\sim,\nicefrac\peq \sim,\nicefrac\rel \sim,\nicefrac V \sim)$ to be such that $\nicefrac W \sim := \{[w]:w\in W\}$, $\nicefrac{\fallible W}\sim := \{[w]:w\in \fallible W\}$, $[w]\nicefrac \peq\sim [v]$ if there exist $w'\sim w$ and $v'\sim v$ such that $w' \peq v' $, $\nicefrac \sqsubseteq \sim$ is the  transitive closure of $\nicefrac {\sqsubseteq^0} \sim$ defined by $[w] \mathrel{\nicefrac {\sqsubseteq^0} \sim} [v]$ whenever there are $w',v' $ so that $w\sim w' \sqsubseteq v' \sim v$, and $ [w] \in \nicefrac V\sim (p)$ iff there is $w'\sim w$ so that $w' \in V(p)$.
\end{defn}

\begin{lem}\label{lemQuotLeq}
	Let $\mathcal M=(W,\fallible W,\peq,\rel,V)$ and suppose that $\sim$ is an equivalence relation that is also a $\Sigma$-bisimulation.
	Then, for all $w,v\in W$, $[w]\mathrel{ \nicefrac \peq\sim }[v]$ if and only if there is $v'\sim v$ such that $ w \peq v' $.
\end{lem}

\begin{proof}
	Clearly, if there is $v'\sim v$ such that $ w \peq v' $, then $[w]\mathrel{ \nicefrac \peq\sim }[v]$.
	Conversely, if $[w]\mathrel{ \nicefrac \peq\sim }[v]$, then there are $w'\sim w$ and $v'\sim v$ such that $w'\peq v'$.
	Since $\sim$ is a bisimulation, there is $v''\sim v'$ such that $w\peq v''$.
	But then $v'' \sim v$, as needed.
\end{proof}

Of particular importance is the case where the relation $\sim$ is given by $\Sigma$-bisimulation.

\begin{lem}\label{lemQuotPreserv}
	Let $\mathcal M=(W,\fallible W,\peq ,\rel ,V )$ be a bi-intuitionistic model and $\Sigma$ be a set of formulas.
	Let $\sim$ be a $\Sigma$-bisimulation which is also an equivalence relation.
	Then,
	\begin{enumerate}
		
		\item If $\mathcal M$ is forward-confluent, so is $\nicefrac{\mathcal M}{\sim }$.
		
		\item If $\mathcal M$ is backward-confluent, so is $\nicefrac{\mathcal M}{\sim }$.
		
		\item If $\mathcal M$ is locally linear, so is $\nicefrac{\mathcal M}{\sim }$.
		
	\end{enumerate}
\end{lem}

\begin{proof}
	For forward confluence, in view of Lemma \ref{lemmClosure}, it suffices to show that $\nicefrac {\sqsubseteq^0} \sim$ is forward confluent.
Assume that $[v] \mathrel{\nicefrac\seq\sim} [w] \mathrel{\nicefrac {\sqsubseteq^0} \sim} [w']$.
By Lemma \ref{lemQuotLeq}, there is $v''\sim v$ such that $v''\seq w$.
Since $\mathcal M$ is forward confluent, there is $v'$ such that $v'' \sqsupseteq v' \seq w'$.
Therefore, $[v] \mathrel{\nicefrac {\sqsupseteq^0} \sim} [v'] \mathrel{\nicefrac\seq\sim} [w']$, as needed.
Backward confluence is treated similarly, and we omit it.
	
For local linearity, assume that $[w] {\nicefrac{\peq}\sim} [u]$ and $[w] {\nicefrac{\peq}\sim} [v]$.
By Lemma~\ref{lemQuotLeq}, let $u'$, $v'$ in $W$ be such that $u \sim u' \seq w \peq v \sim v'$.
Hence, $u' \peq v'$ or $v' \peq u'$. Consecuently, $[u]{\nicefrac{\peq}\sim}[v]$ or $[v]{\nicefrac{\peq}\sim}[u]$.
\end{proof}

\begin{lem}\label{lemQuotTruth}
	If $\mathcal M = (W,\fallible W,\peq,\sqsubseteq,V)$ is any model, $\Sigma$ is closed under subformulas, and ${\sim} \subseteq W\times W$ is a $\Sigma$-bisimulation that is also an equivalence relation, then for all $w \in W$ and $\varphi\in \Sigma$, $(\mathcal M,w)\models\varphi$ iff $(\nicefrac{\mathcal M}\sim,[w])\models\varphi$.
\end{lem}

\begin{proof}
	Proceed by a standard induction on $\varphi$.
	Consider only the interesting cases.
	\medskip

	\noindent {\sc Case} $\varphi = \psi\to \theta$. If $(\mathcal M,w)\models \psi\to \theta$ and $[v] \mathrel {\nicefrac \seq\sim} [w]$, then in view of Lemma \ref{lemQuotLeq} (which we henceforth use without mention), there is $v' \sim v$ such that $w \peq v'$.
	Since $v'\seq w$, either $( \mathcal M,v')\not \models \psi$ or $( \mathcal M,v')   \models \theta$, which by the induction hypothesis and the fact that $[v] = [v']$ yields $(\nicefrac{\mathcal M}\sim,[v])\not \models \psi$ or $(\nicefrac{\mathcal M}\sim,[v]) \models \theta$, and since $[v]  \mathrel {\nicefrac \seq\sim} [w]$ was arbitrary, we conclude that $(\nicefrac{\mathcal M}\sim,[w])\models\varphi$.
	Conversely, if $(\nicefrac{\mathcal M}\sim,[w]) \models \psi \to \theta$ and $v\seq w$, then $[w]\mathrel {\nicefrac \peq\sim} [v]$, which implies that $( \nicefrac{\mathcal M}\sim,[v]) \not \models \psi$ or $(\nicefrac{\mathcal M} \sim,[v]) \models \theta$, and by the induction hypothesis, that $( \mathcal M,v)\not \models \psi$ or $( \mathcal M,v)   \models \theta$, as needed.
	\medskip
	
	\noindent {\sc Case} $\varphi = \ps\psi$.
	Suppose that $(\mathcal M,w) \models \ps \psi $ and let $[v] \mathrel {\nicefrac \seq\sim} [w]$.
	We may assume that $v$ is chosen so that $v\seq w$.
	Then, there is $v'$ such that $v\sqsubseteq v'$ and $(\mathcal M,v') \models \psi $.
	But then $[v] \mathrel{\nicefrac\sqsubseteq\sim} [v']$ and the induction hypothesis yields $(\nicefrac{\mathcal M}\sim,[v']) \models \psi$.
	We conclude that $(\nicefrac{\mathcal M}\sim,[w]) \models \ps \psi$.
Conversely, suppose that $(\mathcal M,w) \not \models \ps \psi $.
Then, there is $v\seq w$ such that $\psi\in \ell^\ps (v)$.
Let $v'$ be such that $[v]\mathrel{ \nicefrac\rel\sim} [v'] $.
Then, there exist sequences $(v_i)_{i\leq k}$ and $(v'_i)_{i\leq n}$ such that
	\[v = v_0 \sim  v'_0 \sqsubseteq  v_1 \sim  v'_1 \sqsubseteq  \ldots \sqsubseteq v_n \sim v'_n = v'. \]
	By induction on $i\leq n$, one readily verifies that $\psi \in \ell^\ps(v_i) \cap \ell^\ps(v'_i)$, from which it follows that $\psi \in \ell^\ps(v')$ and, hence, $  \psi \not \in \ell^+(v')$.
By the induction hypothesis, $(\nicefrac{\mathcal M}\sim,[v']) \not \models \psi$. Since $v \seq w$ then $[v] \mathrel {\nicefrac \seq\sim} [w]$ so $(\nicefrac{\mathcal M}\sim,[w]) \not \models \ps \psi$ : a contradiction.
	\medskip

	\noindent {\sc Case} $\varphi = \nec\psi $. This case is treated very similarly to the previous, but working `dually'. We show only that $(\mathcal M,w) \not \models \nec \psi $ implies that $(\nicefrac{\mathcal M}\sim,[w]) \not \models \nec \psi$ to illustrate.
	If $(\mathcal M,w) \not \models \nec \psi $, there are $v'\sqsupseteq v\seq w$ such that $(\mathcal M,v') \not \models \nec \psi $.
	But then, $ [v'] \mathrel{\nicefrac \sqsupseteq \sim} [v] \mathrel{\nicefrac \seq \sim} [w]$ and the induction hypothesis yields $(\nicefrac{\mathcal M}\sim,[v']) \not \models \psi$, hence $(\nicefrac{\mathcal M}\sim,[w]) \not \models \nec \psi$.
\end{proof}

\begin{lem}\label{lemSingletons}
	Every $\sim_\Sigma$-equivalence class of $\nicefrac{\mathcal M}{\sim_\Sigma}$ is a singleton.
\end{lem}

\begin{proof}
	Suppose that $[w] \sim_\Sigma [v] $; we must show that $ w  \sim_\Sigma  v  $ as well to conclude $[w] = [v]$ (note that $\sim_\Sigma$ is defined both on $\mathcal M$ and $\nicefrac{\mathcal M}\sim$). Define a relation $Z \subseteq W\times W$ given by $x \mathrel Z y$ if $[x] \sim_\Sigma [y]$.
	Clearly $ w\mathrel Z v$, so it remains to check that $Z$ is a $\Sigma$-bisimulation to conclude that $ w \sim_\Sigma v$.
	Clearly $Z$ preserves labels in $\Sigma$.
	We check only the `forth' clause, as the `back' clause is symmetric.
	
	Suppose that $x' \seq   x \mathrel Z y$.
	Since $[x] \sim_\Sigma [y] $, there is $y'$ such that $ [x'] \sim_\Sigma [y'] \mathrel {\nicefrac{\seq}{\sim_\Sigma}} [y]$.
	In view of Lemma \ref{lemQuotLeq}, we may assume that $y'$ is chosen so that $y \peq  y ' $.
	From $[x'] \sim_\Sigma [y']  $ we obtain $x'\mathrel Z y'$, as needed.
\end{proof}

Quotients modulo $\Sigma$-bisimulation will be instrumental in proving the finite model property for $\sf CS4$ and $\sf GS4$.
However, bisimulation does not preserve downward confluence, so to treat $\sf S4I$, we will need a stronger notion of bisimulation.

\begin{defn}
	A {\em strong $\Sigma$-bisimulation} is a $\Sigma$-bisimulation $Z$ such that both $Z$ and $Z^{-1}$ are downward confluent.
We denote the greatest strong $\Sigma$-bisimulation by $\approx_\Sigma$.
\end{defn}

    Notice that an arbitrary union of strong $\Sigma$-bisimulations is a strong $\Sigma$-bisimulation. 
    Hence, $\approx_\Sigma$ is well defined.
    Obviously, $\approx_\Sigma$ is an equivalence relation.
The following is proven in essentially the same way as the forward-confluence preservation clause of Lemma \ref{lemQuotPreserv}.

\begin{lem}\label{lemStrongQuotPreserve}
	Let $\mathcal M $ be a bi-intuitionistic model and $\Sigma$ be a finite set of formulas closed under subformulas.
	Let $\approx$ be a strong $\Sigma$-bisimulation on $\mathcal M$ which is also an equivalence relation.
	Then, if $\mathcal M$ is downward-confluent, so is $\nicefrac{\mathcal M}{\approx}$.
\end{lem}

We remark that the quotient $\nicefrac{\mathcal M}{\sim}$ may be infinite, even taking ${\sim} \in \{ {\sim_\Sigma} , {\approx_\Sigma} \}$.
However, in the next section we discuss a class of models which do have finite quotients.

\section{Shallow Models}\label{sShallow}

A key ingredient in our finite model property proofs is to work with {\em shallow models,} which are models where the {\em heights} of worlds are bounded.
Given a frame $\mathcal F = ( W,\fallible W ,\peq ,\rel )$ and $w\in W$, the {\em height} of $w$ is the supremum of all $n$ such that there is a sequence
\[w = w_0 \prec w_1 \prec \ldots \prec w_n.\]
The height of the frame $\mathcal F  $ is the supremum of all heights of elements of $W$.
Height is defined in the same way if we replace $\mathcal F$ by a model $\mathcal M$.
Note that the height of worlds or models could be $\infty$.
If a frame or model has finite height, we say that it is {\em shallow.}
Shallow models will provide an important intermediate step towards establishing the finite model property, as the bisimulation quotient of a shallow model is finite.
Nevertheless, it can be quite large, as it is only superexponentially bounded.

Below, let $2^x_y$ be the superexponential function given by $2^x_0 = x$ and $2^x_{y+1} = 2^{2^x_y}$.

\begin{lem}
	For all $m,n,k\geq 1$, $2^m\cdot 2_{k}^{(n-1)m}\leq 2_{k}^{nm}.$
\end{lem}

\proof
Proceed by induction on $k$. If $k=1$, then
\[2^m\cdot 2_{k}^{(n-1)m}=2^m\cdot 2^{(n-1)m}=2_{k}^{n m}.\]
If $k>1$, then note that $1 \leq  2_{k-1}^{(n-1)m}$, so that
\[m+ 2_{k-1}^{(n-1)m} \leq (m+1) 2_{k-1}^{(n-1)m} \leq 2^m\cdot 2_{k-1}^{(n-1)m} \stackrel{\text{{\sc IH}}} \leq 2_{k-1}^{nm} .\]
Then,
\[2^m\cdot 2_{k}^{(n-1)m}=2^m \cdot 2^{ 2_{k-1}^{(n-1)m}}=2^{m+ 2_{k-1}^{(n-1)m}} \leq 2^{2_{k-1}^{nm}} = 2_{k}^{nm},\]
as needed.
\endproof

\begin{lem}\label{lemBisBound}
	Given a bi-intuitionistic model $\mathcal M = (W,\fallible W,\peq,\sqsubseteq,V)$ of finite height $n$ and finite $\Sigma$ with $\#\Sigma = s$, $ \# \nicefrac W{\sim_\Sigma} \leq  2^{2(n+1)s}_{n+2}$.
\end{lem}

\proof[Proof sketch]
This is proven in some detail in e.g.~\cite{BalbianiToCL}, but we outline the main elements of the proof.
Proceed by induction on $n\in \mathbb N$ to show that there are at most $2^{2(n+1)s}_{n+2} $ $\Sigma$-bisimulation classes of points of height $n$.
Let $w\in W$.
If $w$ has height $0$, its bisimulation class is uniquely determined by the labels $\ell(v) \in 2^\Sigma\times 2^\Sigma$ of those $v$ in the cluster of $w$ (i.e., the set of $v\in W$ such that $v\peq w \peq v $), and there are at most $2^{2s} = 2^{2s}_1  $ choices for each label, so there are at most $2^{2^{2s}_1} = 2^{2s}_2$ choices for the entire cluster.

For the inductive step, let $\{[v_i]: i\in I\}$ enumerate the equivalence classes of the immediate successors of $w$.
Note that each $v_i$ has height less than that of $w$, so that by the induction hypothesis, there are at most $2^{2ns}_{n+1} $ choices for $[v_i]$, and the bisimulation class of $w$ is determined by the labels of its cluster, for which there are $2^{2^{2s}}$ choices, and a possible choice of  $\{[v_i]: i\in I\}$, of which there are at most $2^{2^{2ns}_{n+1}} $ choices.
Hence there are at most
\[2^{2^{2s}}\cdot 2^{2^{2ns}_{n+1}} = 2^{2^{2s} + 2^{2ns}_{n+1}} \leq 2^{2^{2 (n+1) s}_{n+1}} = 2^{2 (n+1) s}_{n+2} 
\]
choices for the bisimulation class of $w$.
\endproof

The following lemma is proven analogously to Lemma \ref{lemBisBound}.
Below, a bi-intuitionistic model $\mathcal M = (W,\fallible W,\peq,\sqsubseteq,V)$ is {\em forest-like} if for every $w\in W$, the set $\{v\in W:v\peq w\}$ is totally ordered by $\peq$.

\begin{lem}\label{lemStrongBisBound}
	Given a forest-like bi-intuitionistic model $\mathcal M = (W,\fallible W,\peq,\sqsubseteq,V)$ of finite height $n$ and finite $\Sigma$ with $\#\Sigma = s$, $ \# \nicefrac W{\approx_\Sigma} \leq  2^{2(n+1)s}_{n+2}$.
\end{lem}

\begin{remark}
	Note that the forest-like assumption in Lemma \ref{lemStrongBisBound} is needed, as in general there may be infinitely many $\approx_\Sigma$ equivalence classes of points if this assumption fails.
	In a model consisting of an infinite sequence
	\[w_0 \succ v_0 \prec w_1 \succ v_1 \prec w_2 \succ \ldots \]
	where $p$ is true only on $w_0$, no two points are strongly $\sim_{\{p\}}$-bisimilar.
\end{remark}

We conclude this section by showing that the shallow model property implies the finite model property for any of the logics we are interested in.

\begin{thm}\label{thmShallowtoFin}Let $\Lambda\in\! \lbrace{\sf CS4},{\sf IS4},{\sf S4I},{\sf GS4}\rbrace$ and $\varphi\in \lanfull$.
\begin{enumerate}

\item

If $\Lambda \neq\sf S4I$ and $\varphi$ is satisfiable (resp.~falsifiable) in a shallow $\Lambda$-model, then $\varphi$ is satisfiable (resp.~falsifiable) in a finite $\Lambda$-model.

\item

If $\Lambda = \sf S4I$ and $\varphi$ is satisfiable (resp.~falsifiable) in a shallow, forest-like $\Lambda$-model, then $\varphi$ is satisfiable (resp.~falsifiable) in a finite $\Lambda$-model.

\end{enumerate}
\end{thm}

\begin{proof}
Let $\Lambda\in \{{\sf CS4},{\sf IS4},{\sf S4I},{\sf GS4}\}$ and let $\mathcal M$ be a shallow model satisfying (falsifying) $\varphi$.
Then, for $ \Lambda\neq {\sf S4I}$ we see that $\nicefrac{\mathcal M}{\sim_\Sigma}$ is finite by Lemma \ref{lemBisBound}, is a $\Lambda$-model by Lemma \ref{lemQuotPreserv}, and satisfies (falsifies) $\varphi$ by Lemma \ref{lemQuotTruth}.
For $\Lambda = {\sf S4I}$ we further assume that $\mathcal M$ is forest-like, and use Lemma \ref{lemStrongQuotPreserve} to see that $ \nicefrac{\mathcal M}{\approx_\Sigma}$ is downward-confluent.
Reasoning as above, it is also an $\sf S4I$-model satisfying (falsifying) $\varphi$.
\end{proof}

Thus in order to prove the finite model property for any of these logics, it suffices to show that they have the shallow model property: that any non-valid formula is falsifiable in a shallow model.
This is the strategy that we will employ in the sequel.

\begin{remark}
Note that Theorem \ref{thmShallowtoFin} applies to $\sf IS4$, even though we will not establish the finite model property for $\sf IS4$ in this paper.
However, this result does reduce the problem of establishing the finite model property for $\sf IS4$ to that of establishing the shallow model property.
\end{remark}

\section{The finite model property for $\sf CS4$}\label{secFMPCS4}

In view of the above results, in order to prove the finite model property for $\sf CS4$, it suffices to prove the shallow model property.
To this end we define a `shallow' model, $\mathcal M_\Sigma^{\sf CS4}$, which has finite depth.
In this section, the notation $W_c,\peq_c$, etc.~will refer to the canonical model $\mathcal M^{\sf CS4}_c$ for $\sf CS4$.

\begin{defn}
Fix $\Lambda={\sf CS4}$. 
Given $\Gamma, \Delta\in W_c$, define $\Gamma \peq_\Sigma \Delta $ if $\Gamma \peq_c \Delta$ and
\begin{enumerate*}
	\item $\Gamma^+ = \Delta^+$ or
	\item there exists $\chi \in \Sigma $ such that $\chi \in \Delta^+ \setminus \Gamma^+$.
\end{enumerate*}
We then define $\mathcal M_\Sigma^{\sf CS4} = (W_c,\peq_\Sigma,\rel_c,V_c)$.
\end{defn}

Thus $\mathcal M_\Sigma^{\sf CS4}$ is almost identical to $\mathcal M_c^{\sf CS4}$, but we have modified the intuitionistic accessibility relation.
It is easy to see that $\mathcal M_\Sigma^{\sf CS4}$ is a bi-intuitionistic model.

\begin{lem}\label{lemTruthImpCS4}
For all $\Gamma \in W_c$ and $ \varphi \to \psi \in \Sigma $ we have that $\varphi\to \psi \in  \Gamma^+$ iff for all $\Delta \seq_\Sigma \Gamma$, $\varphi\in \Delta^+$ implies $\psi  \in  \Delta^+$.
\end{lem}

\begin{proof}
From left to right, let us take $\Gamma \in W_c$ such that 
$\varphi \rightarrow \psi \in \Gamma^+$. Take any $\Delta \in W_c$ such that $\Gamma \peq_\Sigma \Delta$.
By definition, $\Gamma^+ \subseteq\Delta^+$ so $\varphi \rightarrow \psi \in\Delta^+$. If $\varphi \in \Delta^+$, it follows that $\psi \in \Delta^+$

Conversely, assume towards a contradiction that $\varphi \rightarrow \psi \not \in \Gamma^+$. Therefore, there exists $\Delta \seq_c \Gamma$ such that $\varphi \in \Delta^+$ and $\psi \not \in \Delta^+$: 
\begin{enumerate*}
	\item If $\Delta^+ \cap\Sigma = \Gamma^+\cap\Sigma$ then we already have $\varphi\in \Gamma^+$ and $\psi\not\in \Gamma^+$, and $\Gamma \peq_\Sigma\Gamma$, as needed.
	\item If $\Delta^+ \cap\Sigma \neq \Gamma^+\cap\Sigma$, since $\Delta \seq_c \Gamma$, we conclude that $\Delta\seq_\Sigma \Gamma$.
\end{enumerate*}
\end{proof}

\begin{lem}\label{lemCS4Shallow}
Any $\prec_\Sigma $-chain has length at most $|\Sigma|+1$.
\end{lem}

\begin{proof}
Let $\Gamma_0\prec_\Sigma \Gamma_1 \prec_\Sigma\ldots \prec_\Sigma \Gamma_n$ be any chain.
For each $i<n$ there is $\varphi_i \in \Sigma$ such that $\varphi_i \in \Gamma^+_{i+1} \setminus \Gamma^+_{i }$.
Note that if $i<j<n$ then $\varphi_i\neq\varphi_j$, since by monotonicity $\varphi_i\in  \Gamma^+_{i+1}$ implies that $\varphi_i\in  \Gamma^+_{j}$.
Hence $n\leq |\Sigma|$, so the length of the chain is at most $|\Sigma|+ 1$.
\end{proof}




\begin{lem}\label{lemCS4Back} The model $\mathcal M_\Sigma^{\sf CS4}$ is backward confluent.
\end{lem}

\begin{proof}
If $\Phi \rel_c \Psi \peq_\Sigma \Theta$ then also $\Psi \peq_c \Theta$, so that setting $\Upsilon :=(\Phi^+;\varnothing)$, Proposition \ref{prop:backward}  yields $\Phi \peq_c \Upsilon  \rel_c \Theta$, and clearly we also have $\Phi \peq_\Sigma \Upsilon$, providing the required witness.
\end{proof}

\begin{lem}\label{lemCS4TruthDiam} The following items hold.
\begin{itemize}
	\item If $\ps\varphi \in \Sigma$ and $\Gamma \in W_c$ then $\ps\varphi\in \Gamma^+$ if and only if for all $\Psi\seq_\Sigma \Gamma$ there is $\Delta$ such that $\Psi \rel_c \Delta$ and $ \varphi\in \Delta^+$.
	\item If $\nec\varphi \in \Sigma$ and $\Gamma \in W_c$ then $\nec\varphi\in \Gamma^+$ if and only if for all $\Psi$ and $\Delta$ such that    $ \Gamma \peq_\Sigma \Psi\rel_c \Delta$, $ \varphi\in \Delta^+$.
\end{itemize}	
\end{lem}
\begin{proof}
For the first item, suppose that $\ps\varphi \in \Sigma$ and $\Gamma \in W_c$ is such that $\ps\varphi\in \Gamma^+$.
Let $\Psi \seq_\Sigma \Gamma$.
It follows that $\Psi\seq_c\Gamma$, so  there is $\Delta$ such that $\Psi \rel_c \Delta$ and $ \varphi\in \Delta^+$, as needed.
Conversely, assume that $\ps \varphi \not \in \Gamma^+$.
Consider $\Psi:=(\Gamma^+;\{\varphi\})$.
Then, $\Psi\seq_\Sigma\Gamma$ and if $\Delta \ler_c \Psi$, it follows that $\varphi \not\in \Delta^+$.
   
For the second item, from left to right assume that $\nec\varphi\in \Gamma^+$ and suppose that $\Gamma \peq_\Sigma \Psi \rel_c \Delta$.   
It follows that $\Gamma\peq_c \Psi$, so $\Gamma^+ \subseteq \Psi^+$ and hence $\nec \varphi \in \Psi^+$.
Since $\Psi \rel_c \Delta$, then $\varphi \in \Delta^+$.
Conversely, assume that $\nec \varphi \not \in \Gamma^+$.
As in the proof of Lemma~\ref{lemDiamCS4}, if we define $\Psi:= (\Gamma^+;\varnothing)$, there exists $\Delta\ler_c \Psi$ such that $\varphi \not \in \Delta^+$.
From $\Gamma^+ =\Psi^+$, it is easy to conclude that $\Gamma \peq_\Sigma \Psi$.    
\end{proof}

From Lemmas \ref{lemTruthImpCS4} and \ref{lemCS4TruthDiam} we immediately obtain the following.

\begin{lem}\label{lemTruthCS4}
	For all $\Gamma\in W_c$ and $\varphi\in \Sigma$, $(\mathcal M_\Sigma^{\sf CS4},\Gamma) \models\varphi $ iff $\varphi\in \Gamma$.
\end{lem}

\begin{lem}\label{lemCS4shallowComp}
	For $\varphi\in \Sigma$, $ {\sf CS4}  \vdash\varphi$ if and only if $\mathcal M_\Sigma^{\sf CS4}  \models\varphi$.
\end{lem}

\begin{proof}
	We know that $ {\sf CS4}  \vdash\varphi$ if and only if $ \mathcal M^{\sf CS4}_c \models\varphi$.
	Now, given $\Gamma \in W_c$, and Lemma \ref{lemTruthCS4} yields that $ {\sf CS4}  \vdash\varphi$ implies $(\mathcal M_\Sigma^{\sf CS4} ,\Gamma)   \models\varphi$.
	Conversely, if ${\sf CS4} \not \vdash\varphi$ then there is $\Gamma \in W^{\sf CS4}$ such that $\varphi\not\in \Gamma^+ $, which implies that $\varphi$ $(\mathcal M^{\sf CS4}_\Sigma,  \Gamma )   \not\models\varphi$.
\end{proof}

\begin{thm}\label{thmCS4fin}
	$\sf CS4$ has the finite model property.
\end{thm}

\begin{proof}
	In view of Theorem \ref{thmShallowtoFin}, it suffices to show that $\sf CS4$ has the shallow model property.
	Fix a formula $\varphi$ and let $\Sigma$ be the set of subformulas of $\varphi$. By Lemmas~\ref{lemCS4Shallow},~\ref{lemCS4Back}, $\mathcal M_\Sigma^{\sf CS4} $ is a shallow $\sf CS4$-model, and by Lemma \ref{lemCS4shallowComp}, $\varphi$ is valid on $\mathcal M_\Sigma^{\sf CS4} $ iff ${\sf CS4}\vdash \varphi$, as needed.
\end{proof}

%
%
%
%
%
%
%
%

\section{The finite model property for $\sf GS4$}\label{secGS4fin}

In this section we follow a strategy similar to that of the previous to show that $\sf GS4$ also enjoys the finite model property.
In this case, the worlds of our shallow model will be pairs $(\Gamma,\Delta)$, where the intuition is that $\Gamma$ serves as an `anchor' to enforce local linearity and $\Delta$ represents the formulas of $\Sigma$ true in the given world.
Recall that for $\Theta\in W_c$, we define $\Theta\upharpoonright\Sigma:= (\Theta^+ \cap \Sigma ; \Theta^\ps \cap \Sigma )$.

\begin{defn}
Let $\mathcal M^{\sf GS4}_c = \{W_c,\peq_c,\rel_c,V_c\}$ and $\Sigma\subseteq\lanfull$ be closed under subformulas.
We define $W_\Sigma$ to be the set of pairs $(\Gamma,\Delta)$ where $\Gamma \in W_c$ and $\Delta = \Theta\upharpoonright\Sigma$ for some $\Theta \seq_c \Gamma$.
Set $(\Gamma_0,\Delta_0) \peq _\Sigma (\Gamma_1,\Delta_1)$ iff $\Gamma_0 = \Gamma_1$ and $\Delta_0^+ \subseteq \Delta_1^+ $, and for a propositional variable $p$, set $(\Gamma,\Delta) \in V_\Sigma(p) $ iff $p\in \Delta^+$.
\end{defn}

\begin{lem}\label{lemGS4Model1}
The relation $\peq_\Sigma$ is a locally linear partial order and $V_\Sigma$ is monotone.
\end{lem}

\begin{proof}
It is easy to check that $\peq_\Sigma$ is a partial order, given that $\subseteq$ is, and monotonicity is immediate from the definition of $\peq_\Sigma$.
Suppose that $(\Gamma_0,\Delta_0) \peq _\Sigma (\Gamma_0,\Delta_1)$ and $(\Gamma_0,\Delta_0) \peq _\Sigma (\Gamma_0,\Delta_2)$, and write $\Delta_i =\Theta_i\upharpoonright \Sigma$ with $\Theta_i\seq_c\Gamma_0$.
Then, $\Theta_i\peq_c\Theta_j  $ for some $i\neq j\in \{1,2\}$.
It readily follows that $(\Gamma_0,\Delta_i) \peq _\Sigma (\Gamma_0,\Delta_j)$.
\end{proof}

\begin{lem}\label{lemTruthImpGS4}
Let $\Sigma\subseteq\lanfull$ be closed under subformulas.
	For all $(\Gamma,\Delta) \in W_\Sigma$ and $ \varphi \to \psi \in \Sigma $ we have that $\varphi\to \psi \in  \Delta^+ $ iff for all $(\Gamma,\Psi) \seq _\Sigma (\Gamma,\Delta)$ with $\varphi\in \Psi^+ $ we have $\psi \in  \Psi^+$.
\end{lem}

\begin{proof}
	From left to right, take any $(\Gamma,\Delta) \in W_\Sigma $ satisfying $\varphi \rightarrow \psi \in \Delta^+$. Take any $(\Gamma, \Psi)$ satisfying $(\Gamma,\Psi) \seq _\Sigma (\Gamma,\Delta)$ with $\varphi \in \Psi^+$.
Since $\Delta^+\subseteq \Psi^+$, $\varphi\to \psi \in \Psi^+$, hence $\psi\in \Psi^+$.		
	Conversely, assume that $\varphi \rightarrow \psi \in \Sigma\setminus \Delta^+$.
	By definition, there exists $\Theta \seq_c \Gamma$ such that $\Delta= \Theta \upharpoonright \Sigma$.
	Clearly, $\varphi \rightarrow \psi \not \in \Theta^+$. Therefore there exists $\Upsilon \seq_c \Theta$ such that $\varphi \in \Upsilon^+$ and $\psi \not \in \Upsilon^+$. Define $\Psi = \Upsilon \upharpoonright \Sigma$. We have that $\varphi \in \Psi^+$ and $\psi \not \in \Psi^+$, and clearly $(\Gamma,\Delta) \peq _\Sigma (\Gamma,\Psi )$.
\end{proof}

The following is established by essentially the same reasoning as Lemma \ref{lemCS4Shallow}.

\begin{lem}\label{lemGS4Shallow}
If $\Sigma\subseteq\lanfull$ is finite, then $(W_\Sigma,\peq_\Sigma)$ is shallow.
\end{lem}

\begin{proof}
Any chain $(\Gamma_0,\Delta _0) \prec_\Sigma \ldots \prec_\Sigma (\Gamma_n,\Delta _n) $ satisfies $\Gamma _i = \Gamma_j$ for all $i<j<n$, and hence the elements differ only on the second component.
But a chain $\Delta^+_0\subsetneq \ldots\subsetneq\Delta^+_n$ of subsets of $\Sigma$ can have length at most $|\Sigma| + 1$.
\end{proof}

\begin{defn}
For $(\Gamma,\Delta),(\Gamma',\Delta') \in W_\Sigma$ we set $ (\Gamma,\Delta) \rel^0 _\Sigma (\Gamma',\Delta')$ iff there are $\Theta\seq_c \Gamma$, $\Theta'\seq_c \Gamma'$ such that $\Theta \rel _c \Theta'$, $\Delta = \Theta \upharpoonright \Sigma$, and $\Delta' = \Theta' \upharpoonright \Sigma$.
Then, let $\rel_\Sigma$ be the transitive closure of $\rel^0_\Sigma$.
\end{defn}

\begin{lem}
	The relation $\rel  _\Sigma $ is reflexive and transitive.
\end{lem}
\begin{proof}
Transitivity holds by definition.
For the reflexivity, let us take $(\Gamma,\Delta) \in W_\Sigma$. This means that there exists $\Theta \in W_c$ such that $\Delta = \Theta \upharpoonright \Sigma$ and $\Gamma \peq_c \Theta$. Since $\rel_c$ is reflexive, $\Theta \rel_c \Theta$. By definition, $(\Gamma,\Delta) \rel _\Sigma (\Gamma,\Delta)$.
%
\end{proof}

With this, we define $ \mathcal M^{\sf GS4} _\Sigma  =(W_\Sigma,\peq_\Sigma,\rel_\Sigma,V_\Sigma)$.

\begin{lem}\label{lemGS4Model2}
$\mathcal M^{\sf GS4} _\Sigma $ is forward and backward confluent.
\end{lem}

\begin{proof}
For the forward  condition, we only need to show that $\rel^0 _\Sigma $ is forward confluent.
The claim for $\rel  _\Sigma $ then follows from Lemma \ref{lemmClosure}.
Let us take $(\Gamma_0,\Delta_0)$, $(\Gamma_0,\Delta_2)$ and $(\Gamma_1,\Delta_1)$ in $W_\Sigma$ satisfying $(\Gamma_0,\Delta_0) \peq  _\Sigma  (\Gamma_0,\Delta_2)$ and $(\Gamma_0,\Delta_0) \rel^0 _\Sigma (\Gamma_1,\Delta_1)$.
By definition, there exist $\Theta_0,\Theta_2 \seq_c\Gamma_0$ and $\Theta_1 \seq_c\Gamma_1$ such that $\Delta_i  = \Theta_i \upharpoonright \Sigma$ and $\Theta_0 \rel_c \Theta_1$.
Moreover, we can assume that $\Theta_0\peq_c\Theta_2$, otherwise redefine $\Theta_2:=\Theta_0$.
Since $\mathcal M^{\sf GS4}_c$ is forward confluent, there exists $\Theta_3$ such that $\Theta_2 \rel_c \Theta_3$ and $\Theta_1 \peq_c \Theta_3$. Fix $\Delta_3 = \Theta_3 \upharpoonright \Sigma$. It follows that $(\Gamma_1,\Delta_1) \peq_\Sigma (\Gamma_1, \Delta_3)$ and $(\Gamma_0,\Delta_2) \rel^0 _\Sigma (\Gamma_1,\Delta_3)$.

For the backward  condition, it once again suffices to show that $\rel^0_\Sigma$ is backward confluent.
Let us take $(\Gamma_0,\Delta_0)$, $(\Gamma_1,\Delta_1)$ and $(\Gamma_1,\Delta_2)$ in $W_\Sigma$ satisfying $(\Gamma_0,\Delta_0) \rel ^0_\Sigma  (\Gamma_1,\Delta_1) \peq_\Sigma (\Gamma_1,\Delta_2)$. By definition, there exist $\Theta_0 \seq_c \Gamma_0$ and $\Theta_1 \seq_c \Gamma_1$ such that $\Delta_0 = \Theta_0 \upharpoonright \Sigma$, $\Delta_1  = \Theta_1 \upharpoonright \Sigma$ and  $\Theta_0 \rel_c \Theta_1$. Moreover, there exists $\Theta_2 \seq_c \Gamma_2$ such that $\Delta_2 = \Theta_2 \upharpoonright \Sigma$ and $\Theta_1 \peq_c \Theta_2$. Since $\mathcal M^{\sf GS4}_c$ is backward confluent, there exists $\Theta_3 \in W_c$ such that $\Theta_0 \peq_c \Theta_3 \rel_c \Theta_2$.
Define $\Delta_3 = \Theta_3 \upharpoonright \Sigma$.
It follows that $(\Gamma_0,\Delta_0) \peq _\Sigma (\Gamma_0,\Delta_3) \rel _\Sigma (\Gamma_1,\Delta_2)$.  
\end{proof}

\begin{lem}\label{lemGS4TruthDiam}\
\begin{enumerate}
	\item If $\ps\varphi \in \Sigma$ and $(\Gamma_1,\Delta_1) \in W_\Sigma$ then $\ps\varphi  \in \Delta_1^+$ if and only if there is $(\Gamma_2,\Delta_2) \ler _\Sigma(\Gamma_1,\Delta_1) $ such that $ \varphi\in \Delta_2^+$.
	\item If $\nec\varphi \in \Sigma$ and $(\Gamma_1,\Delta_1) \in W_\Sigma$ then $\nec\varphi \in  \Delta_1^+$ if and only if for every $(\Gamma_2,\Delta_2)  \ler _\Sigma ; \seq _\Sigma(\Gamma_1,\Delta_1) $, $ \varphi\in  \Delta_2^+$.	
\end{enumerate}
\end{lem}

\begin{proof}
For the first item, let us take $(\Gamma_1, \Delta_1) \in W_\Sigma$ satisfying $\ps \varphi \in \Delta_1^+$. By definition, there exists $\Theta_1 \in W_c$ such that $\Delta_1 = \Theta_1 \upharpoonright \Sigma$ and $\Theta_1 \seq_c \Gamma_1$. Therefore, there exists $\Theta_1 \rel_c \Theta_2$ such that $\varphi \in \Theta_2^+$. Define $\Delta_2 = \Theta_2 \upharpoonright \Sigma$. Let us take the pair $(\Theta_2,\Delta_2)$. It can be checked that $(\Gamma_1,\Delta_1) \rel_\Sigma (\Theta_2,\Delta_2)$. Since $\varphi \in \Theta_2^+ \cap \Sigma$ then $\varphi \in \Delta_2^+$.
Conversely, assume that $\ps\varphi \not \in \Delta_1^+$ and suppose that $(\Gamma_1,\Delta_1) \rel^0_\Sigma (\Gamma_2,\Delta_2)$ (we will treat the general case for $\rel_\Sigma$ later).
Then, there exist $\Theta_1\seq_c\Gamma_1 $ and $\Theta_2\seq_c\Gamma_2 $ such that $\Theta_1\rel_c\Theta_2$ and $\Delta_i=\Theta_i\upharpoonright \Sigma$.
From $\Theta_1\rel_c\Theta_2$ we obtain $\varphi\in \Theta_2^\ps$ and hence $\ps\varphi\not\in \Theta_2^+$, so $\ps \varphi\not\in \Delta_2^+$.
Since $\rel_\Sigma$ is the transitive closure of $\rel^0_\Sigma$, we note that $(\Gamma_1,\Delta_1) \rel_\Sigma (\Gamma_2,\Delta_2)$ iff there exists a sequence
\begin{align*}
(\Gamma_1,\Delta_1) & = (\Xi_0,\Upsilon_0) \rel^0_\Sigma (\Xi_1,\Upsilon_1) \rel^0_\Sigma \ldots \\
& \rel^0_\Sigma (\Xi_n,\Upsilon_n) = (\Gamma_2,\Delta_2).
\end{align*}
By induction on $i$, $\ps\varphi\notin\Upsilon_i^+$, so that $\ps\varphi\not\in \Delta_2^+$ and $\varphi\not\in \Delta_2^+$. 
	
Let us consider now the second item.
From right to left, assume by contrapositive that $\nec \varphi \not \in \Delta_1^+$.
By definition, there exists $\Theta_1 \seq_c \Gamma_1$ such that $\Delta_1 = \Theta_1 \upharpoonright \Sigma$. Therefore, $\nec \varphi \not \in \Theta_1^+$. 
As a consequence, there exists $\Theta_1'$ and $\Theta_2$ in $W_\Sigma$ such that $\Theta_1 \peq_c \Theta_1' \rel_c \Theta_2$ and $\varphi \not \in \Theta_2^+$.
Define $\Delta_1' = \Theta_1'\upharpoonright \Sigma$ and $\Delta_2 = \Theta_2\upharpoonright \Sigma$.
By definition $(\Gamma_1,\Delta_1) \peq_\Sigma (\Gamma_1, \Delta_1')$ and $(\Gamma_1,\Delta_1') \rel_\Sigma (\Theta_2,\Delta_2)$. Moreover, since $\varphi \not \in \Theta_2^+$, $\varphi \not \in \Delta_2^+$.  

From left to right, we may argue as in the case for $\ps\varphi$ that if $(\Gamma_1;\Delta_1)\rel_\Sigma (\Gamma_2;\Delta_2)$ and $\nec\varphi\in \Delta^+_1$, then also $\nec\varphi\in \Delta^+_2$, so $\varphi\in \Delta^+_2$.
Thus if $(\Gamma_1;\Delta_1)\peq_\Sigma (\Gamma_3;\Delta_3)\rel_\Sigma (\Gamma_2;\Delta_2)$, from $\nec\varphi\in \Delta^+_1$ we obtain $\nec\varphi\in \Delta^+_3$, and from $(\Gamma_3;\Delta_3)\rel_\Sigma (\Gamma_2;\Delta_2)$ we obtain $\varphi\in \Delta^+_2$.
\end{proof}

From lemmas \ref{lemTruthImpGS4} and \ref{lemGS4TruthDiam} we immediately obtain the following.

\begin{lem}\label{lemTruthGS4}
	Let $(\Gamma,\Delta) \in W_\Sigma$ and $\varphi\in \Sigma$. Then, $(\mathcal M^{\sf GS4}_\Sigma,(\Gamma,\Delta)) \models\varphi $ iff $\varphi\in \Delta^+$.
\end{lem}

\begin{lem}\label{lemGS4hallowComp}
	For $\varphi\in \Sigma$, $ {\sf GS4}  \vdash\varphi$ if and only if $\mathcal M^{\sf GS4}_\Sigma  \models\varphi$.
\end{lem}

\begin{proof}
	We know that $ {\sf GS4}  \vdash\varphi$ if and only if $ \mathcal M^{\sf GS4} _c  \models\varphi$.
	Now, given $(\Gamma,\Delta) \in W_\Sigma$, $\Delta =\Theta\upharpoonright \Sigma$ for some $\Theta \in W_c$, and Lemma \ref{lemTruthGS4} yields that $ {\sf GS4}  \vdash\varphi$ implies that $\varphi\in \Theta^+$, hence $(\mathcal M^{\sf GS4}_\Sigma,(\Gamma,\Delta))   \models\varphi$.
	Conversely, if ${\sf GS4} \not \vdash\varphi$ then there is $\Gamma \in W_c$ such that $\varphi\not\in \Gamma^+ $, which implies that $(\mathcal M^{\sf GS4}_\Sigma, (\Gamma,  \Gamma\upharpoonright \Sigma)) \not  \models\varphi$.
\end{proof}

Reasoning as in the proof of Theorem \ref{thmCS4fin}, we obtain the finite model property for $\sf GS4$.

\begin{thm}
$\sf GS4$ has the finite model property.
\end{thm}


\section{The finite model property for $\sf S4I$}\label{secFMPS4I}

In this section we fix $\Lambda =\sf S4I$.
Our aim is to prove that $\sf S4I$ has the shallow model property, from where the finite model property will follow.
The construction we will use has certain elements in common with those for $\sf CS4 $ and $\sf GS4$, with the caveat that we need to ensure that our models are forest-like in order to apply Theorem \ref{thmShallowtoFin}.
The following construction will ensure this.
In this section, $W_c,\peq_c$, etc.~refer to the components of $\mathcal M^{\sf S4I}_c$, and $W_\Sigma,\peq_\Sigma$, etc.~will refer to the respective components of the model $\mathcal M^{\sf S4I}_\Sigma$ to be constructed below.

\begin{defn}
For a set of formulas $\Sigma$ closed under subformulas, define $W_\Sigma $ to be the set of all tuples $ \Gamma= (\Gamma _0,\ldots,\Gamma _n )$, where
\begin{enumerate*}
\item each $\Gamma _i\in W_c$,
\item $ \Gamma _i \peq_c \Gamma _{i+1}  $ if $i<n$,
\item for each $i<n$ there is a formula $\varphi \in \Sigma$ such that $\varphi \in \Gamma _{i+1}\setminus \Gamma _{i} $.
\end{enumerate*}
We define $V_\Sigma$ by $\Gamma\in V_\Sigma(p)$ iff $p\in  \Gamma_n^+\cap \Sigma$.
Define $\Gamma \peq_\Sigma \Delta$ if $\Gamma$ is an initial sequence of $\Delta$.
\end{defn}

\begin{lem}\label{lemS4IShallow}
The relation $\peq_\Sigma $ is a forest-like partial order on $W_\Sigma$ and any strict $\peq_\Sigma $-chain has length at most $|\Sigma|+1$.
Moreover, $V_\Sigma$ is monotone.
\end{lem}

\begin{proof}
That $\peq_\Sigma $ is a forest-like partial order is easily checked from the definitions, and the bound on chains follows from the same reasoning as for Lemma \ref{lemCS4Shallow}, using the observation that $\Gamma \peq_\Sigma \Delta$ implies that $ \ell^+(\Gamma)\subsetneq \ell^+(\Delta) $.
The monotonicity of $V_\Sigma$ follows from the elements of $\Gamma$ being ordered by $\peq_c$.
\end{proof}

\begin{lem}\label{lemTruthImpS4I}
For all $\Gamma \in W_\Sigma$ and $ \varphi \to \psi \in \Sigma $ we have that $\varphi\to \psi \in \ell^+ ( \Gamma ) $ iff for all $\Delta \seq_\Sigma \Gamma$ with $\varphi\in \ell^+ ( \Delta) $ we have $\psi \in \ell^+( \Delta )$.
\end{lem}

\begin{proof}
First assume that $\Gamma= (\Gamma _0,\ldots,\Gamma _n  )$ is such that $\varphi\to \psi \in \ell^+ ( \Gamma ) $ and let $\Delta = (\Delta _0,\ldots,\Delta _m )$ be such that $\Gamma \peq_\Sigma \Delta$.
Then, $m\geq n$ and $\Delta_n = \Gamma_n$, which by transitivity of $\peq_c$ implies that $\Gamma _n\peq_c \Delta_m$.
It follows that if $\varphi\in \ell^+ ( \Delta) $ then $\psi \in \ell^+( \Delta )$.

Conversely, suppose that $\varphi\to \psi \in \Sigma\setminus \ell^+ ( \Gamma ) $.
We may further assume, without loss of generality, that if $\varphi\in \ell^+(\Gamma)$ then $\psi\in \ell^+(\Gamma)$, for otherwise we may take $\Delta = \Gamma$.
Then, there is $\Delta_{n+1} \seq _c \Gamma_n$ such that $\varphi\in \Delta_{n+1}$ but $\psi \not \in \Delta_{n+1}$.
For $i\leq n$ set $\Delta_i = \Gamma_i$, and define $\Delta = (\Delta_i)_{i\leq n+1}$.
It should then be clear that $\Delta$ has the desired properties.
\end{proof}

The accessibility relation for $\sf S4I$ is a bit more involved than that for $\sf CS4$.

\begin{defn}
For theories $\Phi$, $\Psi$ in $W_c$ we define $\Phi \peq^0_\Sigma \Psi$ if $\Phi \peq_c \Psi$ and $\Phi^+ \cap \Sigma = \Psi^+ \cap \Sigma$.
We define $\Gamma \rel^0_\Sigma \Delta$ if there is $\Theta$ so that $\Gamma\rel_c \Theta \seq^0_\Sigma \Delta$.

We define, for $\Gamma= (\Gamma _0,\ldots,\Gamma _n  )$ and $\Delta = (\Delta _0,\ldots,\Delta _m  ) \in W_\Sigma$, $\Gamma \rel^1_\Sigma \Delta$ if there is a non-decreasing sequence $j_1,\ldots j_n $ with $j_n = m$ such that $ \Gamma _i  \rel^0_\Sigma  \Delta _{j_i} $ for $i\leq n$.
We define $\rel_\Sigma$ to be the transitive closure of $\rel^1_\Sigma$.
\end{defn}

With this, we define $\mathcal M^{\sf S4I}_\Sigma = (W_\Sigma,\peq_\Sigma,\rel_\Sigma,V_\Sigma)$.
It is easy to check using the above lemmas that $\mathcal M^{\sf S4I}_\Sigma$ is a bi-intuitionistic model.
Next we show that it is indeed an $\sf S4I$ model.

\begin{lem}\label{lemIsS4IModel}
The relation $\rel_\Sigma $ is forward and downward confluent on $ S$.
\end{lem}

\begin{proof}
As before, it suffices to show that $\rel^1_\Sigma $ is forward and downward confluent.
Suppose that $\Phi \peq_\Sigma \Phi' $ and $\Phi \rel^1_\Sigma \Delta$ and write $\Phi' = (\Phi_i)_{i\leq m'}$, so that $\Phi=(\Phi_i)_{i\leq m}$ for some $m\leq m'$, and $\Delta = (\Delta_i)_{i\leq n}$.
We construct $\Delta' = (\Delta_i)_{i\leq n'}$ for some $n'\geq n$ such that $\Phi' \rel_\Sigma \Delta'$ and $\Delta \peq_\Sigma \Delta' $.
Given $k \in [m,m']$, we assume inductively that $(\Delta_i)_{i\leq r}$ have been built so that $(\Phi_i)_{i\leq k} \rel_\Sigma (\Delta_i)_{i\leq r}  $.
The base case with $r=m$ is already given by the assumption that $\Phi \rel_\Sigma  \Delta $.
For the inductive step, assume that $(\Phi_i)_{i\leq k} \rel_\Sigma (\Delta_i)_{i\leq r}  $.
Then in particular, $\Phi_k \rel^0_\Sigma \Delta_r$.
By the definition of $\rel_\Sigma^0$, there is $\Theta \in W_c$ so that $\Phi_k \rel_c \Theta \seq^0_\Sigma \Delta_r$.
By forward confluence of $\mathcal M^{\sf S4I}_c$, there is $\Upsilon$ such that $\Phi_{k+1} \sqsubseteq_c \Upsilon \seq_c \Theta$.
If $ \Delta_r^+ \cap \Sigma =  \Upsilon^+ \cap \Sigma$, we observe that $\Delta_r \peq^0_\Sigma \Upsilon$.
Hence $\Phi_{k+1} \sqsubseteq^0_\Sigma \Delta_r$, so that $(\Phi_i)_{i\leq k + 1} \rel^1_\Sigma (\Delta_i)_{i\leq r}  $.
In this case, we may simply set $r'=r $.
Otherwise, we set $r'=r+1$  and $\Delta_{r+1} = \Upsilon$.
In this case, we also have that $(\Phi_i)_{i\leq k} \rel^1_\Sigma (\Delta_i)_{i\leq r}  $.
\end{proof}

\begin{lem}\label{lemTruthDiamS4I}
Let $\Gamma \in W_\Sigma$.
\begin{enumerate}

\item If $\ps\varphi \in \Sigma$ then $\ps\varphi \in \ell^+( \Gamma ) $ if and only if there is $\Delta \ler _\Sigma\Gamma $ such that $ \varphi\in \ell^+( \Delta )$.

\item If $\nec\varphi \in \Sigma$ then $\nec\varphi \in \ell^+( \Gamma ) $ if and only if for every $\Delta \ler _\Sigma\Gamma $, $ \varphi\in  \ell^+(  \Delta )$.

\end{enumerate}

\end{lem}

\begin{proof}
Let $\Gamma= (\Gamma_i)_{i\leq n}$.
For the first claim, we first prove the easier right-to-left direction.
Suppose that $\Delta = (\Delta_i)_{i\leq m} \ler^1 _\Sigma\Gamma $ is such that $ \varphi\in \ell^+( \Delta )$, meaning that $\varphi \in \Delta_m^+$.
Then, $\Gamma_n \rel^0_\Sigma \Delta_m $, which means that for some $\Theta$, $\Gamma_n \rel_c \Theta \seq^0_\Sigma \Delta_m $.
Then $\varphi \in \Theta^+$, hence $\ps\varphi \in \Gamma_n^+ = \ell^+(\Gamma)$ by the truth lemma.
If instead $\Delta\ler_\Sigma \Gamma$, then there is a sequence
\[\Gamma = \Upsilon_0 \rel^1_\Sigma \Upsilon_1 \rel^1_\Sigma \ldots \rel^1_\Sigma \Upsilon_n = \Delta,\]
and backward induction on $i$ shows that $\ps\varphi\in \ell^+(\Upsilon_i)$, so that in particular $\ps\varphi\in \ell^+(\Gamma)$.

For the left-to-right direction, suppose that $\ps\varphi \in  \ell^+( \Gamma ) $.
By backward induction on $k\leq n$ we prove that there exists some sequence $ (\Theta_i)_{i = k}^n$ such that $ \Gamma_i \rel _c \Theta_i $ for each $i\in [k,n]$ and $\Theta_i\peq_c\Theta_{i+1}$ if $i\in[k,n)$ (interval notation should be interpreted over the natural numbers).
In the base case $k=n$ and, by the truth lemma, there is $\Theta_n \ler_c \Gamma_n$ such that $\varphi\in \Theta_n$.
So, suppose that $ (\Theta_i)_{i = k+1}^n$ has been constructed with the desired properties.
Since the canonical model is downward confluent, there is some $\Theta_k$ such that $\Gamma_k \rel_c\Theta_k\peq_c \Theta_{k+1} $, yielding the desired $\Theta_k$.

The problem is that the sequence $\Theta=(\Theta_i)_{i\leq n}$ may not be an element of $W_\Sigma$.
We instead choose a suitable subsequence $\Delta = (\Theta_{j_i})_{i\leq m}$, where $j_k$ is defined by (forward) induction on $k$.
For the base case, we set $j_0 = 0$, and $m_0 = 0$ (meaning that $\Delta$ currently has $m_0+1$ elements).
Now, suppose that $m_k$ has been defined as has been $j_i$ for $i\leq m_k$, in such a way that $(\Theta_{j_i})_{i \leq m_k} \in  W_\Sigma$ and $(\Gamma_i)_{i\leq k} \rel_\Sigma (\Theta_{j_i})_{i \leq m_k}$.
To define $m_{k+1} \geq m_k$ and $(\Theta_{j_i})_{i \leq m_{k+1}}$, we consider two cases.
First assume that there is $\psi\in \Sigma$ such that $ \psi \in  \Theta_{k+1} \setminus \Theta_{j_k}  $.
In this case, setting $m_{k+1} = m_k + 1$ and $j_{m_{k+1}} = k+1 $ we see that $(\Theta_{j_i})_{i \leq m_{k+1}}$ has all desired properties.
In particular, the existence of the formula $\psi$ guarantees that $(\Theta_{j_i})_{i \leq m_{k+1}} \in W_\Sigma$.
Otherwise, set $m_{k+1} = m_k$.
In this case we see that $\Gamma_{k+1} \rel_c \Theta_{k+1} \seq^0_\Sigma \Theta_{j_{m_k}}$, so that $\Gamma_{k+1} \rel^0_\Sigma \Theta_{j_{m_k}}$ and thus $(\Gamma_i)_{i\leq k+1} \rel_\Sigma (\Theta_{j_i})_{i \leq m_{k+1}}$.
That $(\Theta_{j_i})_{i \leq m_{k+1}} \in W_\Sigma$ follows from the fact that $m_{k+1} = m_k$ and the induction hypothesis.
The claim then follows by setting $m = m_n$ and $\Delta = (\Theta_{j_i})_{i \leq m  }$.

The second claim is proven similarly, but by contrapositive. We leave the details to the reader.
\end{proof}

From Lemmas \ref{lemTruthImpS4I} and \ref{lemTruthDiamS4I} we immediately obtain the following.

\begin{lem}\label{lemTruthS4I}
For $\Gamma \in W_\Sigma$ and $\varphi\in \Sigma$, $(\mathcal M^{\sf S4I}_\Sigma,\Gamma) \models\varphi $ iff $\varphi\in \ell^+(\Gamma)$.
\end{lem}

\begin{lem}\label{lemS4IShallowComp}
For $\varphi\in \Sigma$, $ {\sf S4I}  \vdash\varphi$ if and only if $\mathcal M^{\sf S4I}_\Sigma  \models\varphi$.
\end{lem}

\begin{proof}
We know that $ {\sf S4I}  \vdash\varphi$ if and only if $ \mathcal M^{\sf S4I}   \models\varphi$.
Now, given $\Gamma \in W_\Sigma$, $\ell^+(\Gamma)=\Theta \cap \Sigma$, for some prime set $\Theta$ containing all derivable formulas, so $\varphi \in \ell^+(\Gamma)$. Lemma \ref{lemTruthS4I} yields that $ {\sf S4I}  \vdash\varphi$ implies $(\mathcal M^{\sf S4I}_\Sigma,\Gamma)   \models\varphi$.
Conversely, if ${\sf S4I} \not \vdash\varphi$ then there is $\Gamma \in W_c$ such that $\varphi\not\in \Gamma $, which setting $\Gamma'\! =\!(\Gamma)$ to be a singleton sequence yields $(\mathcal M^{\sf S4I}_\Sigma, \Gamma')   \not \models\varphi$.
\end{proof}

\begin{thm}
$\sf S4I$ has the finite model property.
\end{thm}

\begin{proof}
In view of Theorem \ref{thmShallowtoFin}, it suffices to show that $\sf S4I$ has the shallow, forest-like model property.
Fix a formula $\varphi$ and let $\Sigma$ be the set of subformulas of $\varphi$.
By Lemmas \ref{lemS4IShallow} and \ref{lemIsS4IModel}, $\mathcal M^{\sf S4I}_\Sigma$ is a shallow, forest-like $\sf S4I$-model, and by Lemma \ref{lemS4IShallowComp}, $\varphi$ is valid on $\mathcal M^{\sf S4I}_\Sigma$ iff ${\sf S4I}\vdash \varphi$, as needed.
\end{proof}

\section{Concluding remarks}

We have settled the long-standing problem of the finite model property for $\sf CS4$ and introduced two logics closely related to $\sf IS4$ which also enjoy the finite model property.
The logics we considered correspond to classes of models with a combination of the forward, backward and downward confluence properties.
There are a handful of other logics which may be defined in this fashion.
For example, consider the class of downward confluent frames.
While the axiom \ref{ax:cd} holds on frames which are forward {\em and} downward confluent, it is unclear how the class of downward confluence frames may be axiomatized.
On the other hand, it seems that the shallow model construction we have provided for $\sf S4I$ should readily adapt to frames satisfying only downward confluence.
In total, there are 8 logics that could be obtained by combining these conditions, and 8 more for their locally linear variants.

However, we have focused specifically on $\sf GS4$ and $\sf S4I$ because of their similarities with $\sf IS4$.
We expect that this work will shed some light on the challenges and possible strategies towards establishing the finite model property for $\sf IS4$.

The finite model property for $\sf GS4$ suggests that a proof that $\sf IS4$ does {\em not} have the finite model property could not be obtained from a straightforward adaptation of the proof that classical expanding products do not enjoy the finite {\em expanding model} property, i.e.~the finite model property with respect to their `intended' class of frames.
$\sf GS4$ frames are closely related to the expanding product ${\sf S4.3} \times^e {\sf S4}$ and $\sf IS4$ to ${\sf S4} \times^e {\sf S4}$.
Both of these classes of expanding products lack the finite expanding model property, as does ${\sf S4.3} \times^e {\sf S4.3}$; the same proof is used to establish these three facts, hence it does not involve branching of either accessibility relation \cite{pml}.
This suggests that a proof that $\sf IS4$ lacks the FMP would need a new construction which uses the non-linearity of $\peq$ in an essential way.

On the other hand, the finite model property for $\sf S4I$ might suggest that a similar approach may lead to the FMP for $\sf IS4$.
However, one should be careful, as the role of the accessibility relations $\peq$, $\sqsubseteq$ in the proof are not at all symmetric.
We have reduced the finite model property to the shallow model property, but it is not clear how one should define a shallow $\sf IS4$ model from an arbitrary one; in the absence of linearity, we have provided one construction that preserves forward confluence and one that preserves backward confluence.
Preserving both conditions seems to be much more difficult, if at all possible.

Finally, we mention that G\"odel modal logics are of independent interest (see e.g.~\cite{Caicedo2010,metcalfe}), with the G\"odel variants of the modal logics ${\sf K}$ and ${\sf S5}$ enjoying the finite model property~\cite{CaicedoMRR13}.
The techniques we have developed here may be adapted to treating other G\"odel modal logics.
One particularly interesting case study may be G\"odel linear temporal logic, in the spirit of intuitionistic temporal logic \cite{BalbianiToCL}.
\section*{Acknowledgments}
David Fern\'andez-Duque's research is partially funded by the SNSF-FWO Lead Agency Grant 200021L\_196176 (SNSF)\slash G0E2121N (FWO). Martin Di\'eguez's research is partially funded by the research project Amorcage/EL4HC of the program RFI Atlanstic 2020.

\end{document}